\documentclass{article}

\usepackage[utf8]{inputenc} 
\usepackage[T1]{fontenc}    
\usepackage[colorlinks=true,linkcolor=black,anchorcolor=black,citecolor=black,filecolor=black,menucolor=black,runcolor=black,urlcolor=black]{hyperref}       
\usepackage{url}            
\usepackage{booktabs}       
\usepackage{amsfonts}       
\usepackage{nicefrac}       
\usepackage{microtype}     
\usepackage{xcolor}         
\usepackage{bm}         
\usepackage{graphicx}  		
\usepackage{amsmath} 
\usepackage{cite} 
\usepackage{diagbox}
\usepackage{enumitem}
\usepackage{siunitx}
\usepackage{multirow}
\usepackage{tabularx}
\usepackage{hhline}
\usepackage{arydshln}		
\usepackage{xcolor}
\usepackage{mathtools} 
\usepackage{amsthm}			
\usepackage{amssymb}			
\usepackage{algorithmic}	
\usepackage{algorithm}
\usepackage{pifont}
\usepackage{threeparttable}
\usepackage{fullpage}

\newcommand{\cmark}{\ding{51}}%
\newcommand{\xmark}{\ding{55}}%
\definecolor{lightgreen}{rgb}{.9,1,.9}
\definecolor{lightred}{rgb}{1,.415,.415}
\definecolor{lightblue}{rgb}{.415,.415,1}

\newcolumntype{L}[1]{>{\raggedright\arraybackslash}p{#1}}
\newcolumntype{C}[1]{>{\centering\arraybackslash}p{#1}}
\newcolumntype{R}[1]{>{\raggedleft\arraybackslash}p{#1}}

\theoremstyle{plain} 

\newtheorem{theorem}{Theorem}
\newtheorem{lemma}{Lemma}

\newtheorem{assumption}{Assumption}

\newtheorem*{suptheorem}{Theorem}

\def\defn{\,\coloneqq\,}

\def\diag{{\mathsf{diag}}}
\def\d{{\mathsf{\, d}}}
\def\Zer{{\mathsf{Zer}}}
\def\mod{{\mathsf{mod}}}
\def\Im{{\mathsf{Im}}}
\def\prox{{\mathsf{prox}}}

\def\max{{\mathsf{max}}}
\def\min{\mathop{\mathsf{min}}}
\def\Lmax{{L_{\mathsf{\tiny max}}}}
\def\Mmax{{M_{\mathsf{\tiny max}}}}
\def\asarrow{{\,\xrightarrow{\mathsf{\tiny a.s.}}\,}}

\def\C{\mathbb{C}}
\def\R{\mathbb{R}}
\def\E{\mathbb{E}}

\def\zerobm{{\bm{0}}}

\def\ebm{{\bm{e}}}

\def\xbm{{\bm{x}}}
\def\zbm{{\bm{z}}}
\def\ybm{{\bm{y}}}
\def\zbm{{\bm{z}}}

\def\nbm{{\bm{n}}}
\def\ubm{{\bm{u}}}
\def\vbm{{\bm{v}}}

\def\varepsilonbar{{\overline{\varepsilon}}}

\def\Abm{{\bm{A}}}
\def\Dbm{{\bm{D}}}
\def\Pbm{{\bm{P}}}
\def\Fbm{{\bm{F}}}

\def\sigmabm{{\bm{\sigma}}}

\def\thetabm{{\bm{\theta }}}

\def\Ubf{{\mathbf{U}}}
\def\Ibf{{\mathbf{I}}}

\def\Abm{{\bm{A}}}
\def\Dbm{{\bm{D}}}
\def\Pbm{{\bm{P}}}
\def\Fbm{{\bm{F}}}

\def\Fcal{{\mathcal{F}}}

\def\Lcal{{\mathcal{L}}}

\def\Dsf{{\mathsf{D}}}
\def\Gsf{{\mathsf{G}}}

\def\Tsf{{\mathsf{T}}}

\def\Tsf{{\mathsf{T}}}

\def\Dsf{{\mathsf{D}}}

\def\Hsf{{\mathsf{H}}}
\def\Jsf{{\mathsf{J}}}

\def\Gsf{{\mathsf{G}}}

\def\Hsf{{\mathsf{H}}}

\def\xbmhat{{\widehat{\bm{x}}}}
\def\zbmhat{{\widehat{\bm{z}}}}

\def\argmin{\mathop{\mathsf{arg\,min}}} 

\newcommand{\norm}[1]{\left\lVert#1\right\rVert}

\title{Block Coordinate Plug-and-Play Methods\\for Blind Inverse Problems}
\date{}

\author{
Weijie Gan, Shirin Shoushtari, Yuyang Hu, Jiaming Liu, \\ Hongyu An, and Ulugbek S.\ Kamilov\\
\small Washington University in St. Louis, MO 63130, USA\\
\small \texttt{\{weijie.gan, s.shirin, h.yuyang, jiaming.liu, hongyuan, kamilov\}@wustl.edu}
}

\begin{document}

\maketitle

\begin{abstract}
Plug-and-play (PnP) prior is a well-known class of methods for solving imaging inverse problems by computing fixed-points of operators combining physical measurement models and learned image denoisers. While PnP methods have been extensively used for image recovery with known measurement operators, there is little work on PnP for solving blind inverse problems. We address this gap by presenting a new block-coordinate PnP (BC-PnP) method that efficiently solves this joint estimation problem by introducing learned denoisers as priors on both the unknown image and the unknown measurement operator. We present a new convergence theory for BC-PnP compatible with blind inverse problems by considering \emph{nonconvex} data-fidelity terms and \emph{expansive} denoisers. Our theory analyzes the convergence of BC-PnP to a stationary point of an \emph{implicit} function associated with an \emph{approximate} minimum mean-squared error (MMSE) denoiser. We numerically validate our method on two blind inverse problems: automatic coil sensitivity estimation in magnetic resonance imaging (MRI) and blind image deblurring. Our results show that BC-PnP provides an efficient and principled framework for using denoisers as PnP priors for jointly estimating measurement operators and images.
\end{abstract}


\section{Introduction}

Many problems in computational imaging, biomedical imaging, and computer vision can be formulated as \emph{inverse problems} involving the recovery of high-quality images from low-quality observations. Imaging inverse problems are generally ill-posed, which means that multiple plausible clean images could lead to the same observation. It is thus common to introduce prior models on the desired images. While the literature on prior modeling of images is vast, current methods are often based on \emph{deep learning (DL)}, where a deep model is trained to map observations to images~\cite{McCann.etal2017, Lucas.etal2018, Ongie.etal2020}. 

\emph{Plug-and-play (PnP) priors}~\cite{Venkatakrishnan.etal2013, Sreehari.etal2016} is one of the most widely-used DL frameworks for solving imaging inverse problems. PnP methods circumvent the need to explicitly describe the full probability density of images by specifying image priors using image denoisers. The integration of state-of-the-art deep denoisers with physical measurement models within PnP has been shown to be effective in a number of inverse problems, including image super-resolution, phase retrieval, microscopy, and medical imaging~\cite{Metzler.etal2018, Zhang.etal2017a, Meinhardt.etal2017, Dong.etal2019, Zhang.etal2019, Wei.etal2020, Zhang.etal2022, Liu.etal2022} (see also recent reviews~\cite{Ahmad.etal2020, Kamilov.etal2023}). Practical success of PnP has also motivated novel extensions, theoretical analyses, statistical interpretations, as well as connections to related approaches such as score matching and diffusion models~\cite{Chan.etal2016, Romano.etal2017, Buzzard.etal2017, Reehorst.Schniter2019, Sun.etal2018a, Sun.etal2019b, Ryu.etal2019, Xu.etal2020, Liu.etal2021b, Kadkhodaie.Simoncelli2021, Cohen.etal2021a, Hurault.etal2022, hurault2022proximal, Laumont.etal2022}.

Despite the rich literature on PnP, the existing work on the topic has primarily focused on the problem of image recovery where the measurement operator is known exactly. There is little work on PnP for \emph{blind} inverse problems, where both the image and the measurement operator are unknown. This form of blind inverse problems are ubiquitous in computational imaging with well-known applications such as blind deblurring~\cite{campisi2017blind} and parallel magnetic resonance imaging (MRI)~\cite{Fessler2020}. In this paper, we address this gap by developing a new PnP approach that uses denoisers as priors over both the unknown measurement model and the unknown image, and efficiently solves the joint estimation task as a \emph{block-coordinate PnP (BC-PnP)} method. While a variant of BC-PnP was proposed in the recent paper~\cite{Sun.etal2019b}, it was never used for jointly estimating the images and the measurement operators. Additionally, the convergence theory in~\cite{Sun.etal2019b} is inadequate for blind inverse problems since it assumes convex data-fidelity terms and nonexpansive denoisers. We present a new convergence analysis applicable to \emph{nonconvex} data-fidelity terms and \emph{expansive} denoisers. Our theoretical analysis provides explicit error bounds on the convergence of BC-PnP for \emph{approximate} minimum mean squared error (MMSE) denoisers under a set of clearly specified assumptions. We show the practical relevance of BC-PnP by solving joint estimation problems in blind deblurring and accelerated parallel MRI. Our numerical results show the potential of denoisers to act as PnP priors over the measurement operators as well as images. Our work thus addresses a gap in the current PnP literature by providing a new efficient and principled framework applicable to a wide variety of blind imaging inverse problems.

All proofs and some details that have been omitted for space appear in the supplementary material.

\section{Background}

\textbf{Inverse Problems.} Many imaging problems can be formulated as inverse problems where the goal is to estimate an unknown image $\xbm\in\R^n$ from its degraded observation $\ybm=\Abm\xbm + \ebm$, where $\Abm\in\R^{m\times n}$ is a measurement operator and $\ebm\in\R^m$ is the noise. A common approach for solving inverse problems is based on formulating an optimization problem
\begin{equation}
    \label{equ:optimization}
    \xbmhat \in \argmin_{\xbm\in\R^n} f(\xbm) \quad\text{with}\quad f(\xbm) = g(\xbm) + h(\xbm)\ ,
\end{equation}
where $g$ is the data-fidelity term that quantifies consistency with the observation $\ybm$ and $h$ is the regularizer that infuses a prior on $\xbm$. For example, a widely-used data-fidelity term and regularizer in computational imaging are the least-squares $g(\xbm)=\frac{1}{2}\norm{\Abm\xbm-\ybm}_2^2$ and the total variation (TV) functions $h(\xbm)=\tau\norm{\Dbm\xbm}_1$, where $\Dbm$ is the image gradient, and $\tau > 0$ a regularization parameter. 

The traditional inverse problem formulations assume that the measurement operator $\Abm$ is known exactly. However, in many applications, it is more practical to model the measurement operator as $\Abm(\thetabm)$, where $\thetabm \in \R^p$ are unknown parameters to be estimated jointly with $\xbm$. This form of inverse problems are often referred to as \emph{blind} inverse problems and arise in a wide-variety of applications, including pralellel MRI~\cite{Ying.Sheng2007, Uecker.etal2008, Jun.etal2021, Holme.etal2019, Arvinte.etal2021, Sriram.etal2020}, blind deblurring~\cite{kundur1996blind, campisi2017blind, pan2017deblurring}, and computed tomography~\cite{malhotra2016tomographic,lee2017phantomless,cheng2018correction,Xie.etal2021}.

\textbf{DL.} There is a growing interest in DL for solving imaging inverse problems~\cite{McCann.etal2017, Lucas.etal2018, Ongie.etal2020}. Instead of explicitly defining a regularizer, DL approaches for solving inverse problems learn a mapping from the measurements to the desired image by training a convolutional neural network (CNN) to perform regularized inversion~\cite{Wang2016.etal, DJin.etal2017, Kang.etal2017, Chen.etal2017, Xu.etal2018a}. Model-based DL (MBDL) has emerged as powerful DL framework for inverse problems that combines the knowledge of the measurement operator with an image prior specified by a CNN (see reviews~\cite{Ongie.etal2020, Monga.etal2021}). The literature of MBDL is vast, but some well-known examples include PnP, regularization by denoising (RED), deep unfolding (DU), compressed sensing using generative models (CSGM), and deep equilibrium models (DEQ)~\cite{Bora.etal2017, zhang2018ista, Hauptmann.etal2018, Gilton.etal2021, Liu.etal2022a}. All these approaches come with different trade-offs in terms of imaging performance, computational and memory complexity, flexibility, need for supervision, and theoretical understanding.

The literature on DL approaches for blind inverse problems is broad, with many specialized methods developed for different applications. While an in-depth review would be impractical for this paper, we mention several representative approaches adopted in prior work. The direct application of DL to predict the measurement operator from the observation was explored in~\cite{han.etal2018, peng2022deepsense}. Deep image prior (DIP) was used as a learning-free prior to regularize the image and the measurement operator in~\cite{bostan2020deep,ren2020neural}. Generative models, including both GANs and diffusion models, have been explored as regularizers for blind inverse problems in~\cite{asim2020blind, Chung.etal2023a}. Other work considered the use of a dedicated neural network to predict the parameters of the measurement operator, adoption of model adaptation strategies, and development of autocalibration methods based on optimization~\cite{Jun.etal2021, Holme.etal2019, Arvinte.etal2021, Sriram.etal2020, hu2022spice, Gossard.Weiss2022, huang2022unrolled,li2020efficient}.

\textbf{PnP.} PnP~\cite{Venkatakrishnan.etal2013, Sreehari.etal2016} is one of the most popular MBDL approaches based on using deep denoisers as imaging priors (see also recent reviews~\cite{Ahmad.etal2020, Kamilov.etal2023}). For example, the proximal gradient method variant of PnP (referred to as PnP-ISTA in this paper) can be formulated as a fixed-point iteration~\cite{Kamilov.etal2017}
\begin{equation}
    \label{equ:pnp}
    \xbm^{k} \leftarrow \Dsf_\sigma\left(\zbm^k\right) \quad\text{with}\quad \zbm^k \leftarrow \xbm^{k-1} - \gamma\nabla g(\xbm^{k-1})\ ,
\end{equation}
where $\Dsf_\sigma$ is a denoiser with a parameter $\sigma > 0$ for controlling its strength and $\gamma > 0$ is a step-size. The theoretical convergence of PnP-ISTA has been explored for convex functions $g$ using monotone operator theory~\cite{Sun.etal2018a, Ryu.etal2019} as well as for nonconvex functions based on interpreting the denoiser as a MMSE estimator~\cite{Xu.etal2020}. The analysis in this paper builds on the convergence theory in~\cite{Xu.etal2020} that uses an elegant formulation by Gribonval~\cite{Gribonval2011} establishing a direct link between MMSE estimation and regularized inversion. Many variants of PnP have been developed over the past few years~\cite{Metzler.etal2018, Zhang.etal2017a, Meinhardt.etal2017, Dong.etal2019, Zhang.etal2019, Wei.etal2020, Zhang.etal2022}, which has motivated an extensive research on its theoretical properties~\cite{Chan.etal2016, Buzzard.etal2017, Ryu.etal2019, Sun.etal2018a, Tirer.Giryes2019, Teodoro.etal2019, Xu.etal2020, Sun.etal2021, Cohen.etal2020, Hurault.etal2022, Laumont.etal2022, hurault2022proximal}.

\emph{Block coordinate regularization by denoising (BC-RED)} is a recent PnP variant for solving large-scale inverse problems by updating along a subset of coordinates at every iteration~\cite{Sun.etal2019b}. BC-RED is based on regularization by denoising (RED), another well-known variant of PnP that seeks to formulate an explicit regularizer for a given image denoiser~\cite{Romano.etal2017, Reehorst.Schniter2019}. BC-RED was applied to several non-blind inverse problems and was theoretically analyzed for convex data-fidelity terms. 

PnP was extended to blind deblurring in~\cite{ljubenovic2017blind, ljubenovic2019plug} by considering an additional prior on blur kernels that promotes sparse and nonegative solutions. PnP was also applied to holography with unknown phase errors by using the Gaussian Markov random field model as the prior for the phase errors~\cite{pellizzari2020coherent}. \emph{Calibrated RED (Cal-RED)}~\cite{Xie.etal2021} is a recent related extension of RED that calibrates the measurement operator during RED reconstruction by combining the traditional RED updates over an unknown image with a gradient descent over the unknown parameters of the measurement operator. However, this prior work does not leverage any learned priors for the measurement operator and does not provide any theoretical analysis.

\textbf{Our contributions.} \textbf{\emph{(1)}} Our first contribution is in the use of learned deep denoisers for regularizing the measurement operators within PnP. While the idea of calibration within PnP was introduced in~\cite{Xie.etal2021}, denoisers were not used as priors for measurement operators. \textbf{\emph{(2)}} Our second contribution is the application of BC-PnP as an efficient method for jointly estimating the unknown image and the measurement operator. While BC-RED was introduced in~\cite{Sun.etal2019b} as a block-coordinate variant of PnP, the method was used for solving non-blind inverse problems by using patch-based image denoisers. \textbf{\emph{(3)}} Our third contribution is a new convergence theory for BC-PnP for the \emph{sequential} and \emph{random} block-selection strategies under approximate MMSE denoisers. Our analysis does \emph{not} assume convex data-fidelity terms, which makes it compatible with blind inverse problems. Our analysis can be seen as an extension of~\cite{Xu.etal2020} to block-coordinate updates and approximate MMSE denoisers. \textbf{\emph{(4)}} Our fourth contribution is the implementation of BC-PnP using learned deep denoisers as priors for two distinct blind inverse problems: blind deblurring and auto-calibrated parallel MRI. Our code---which we share publicly---shows the potential of learning deep denoisers over measurement operators and using them for jointly estimating the uknown image and the uknown measurement operator.

\section{Block Coordinate Plug-and-Play Method}

We propose to efficiently solve blind inverse problems by using a block-coordinate PnP method, where each block represents one group of unknown variables (images, measurement operators, etc). The novelty of our work relative to~\cite{Sun.etal2019b} is in solving blind inverse problems by introducing learned priors on both the unknown image and the uknown measurement operator. Additionally, unlike~\cite{Sun.etal2019b}, our work proposes a fully nonconvex formulation that is more applicable to blind inverse problems.

Consider the decomposition of a vector $\xbm \in \R^n$ into $b \geq 1$ blocks
\begin{equation}
\xbm = (\xbm_1, \cdots, \xbm_b) \in \R^{n_1} \times \cdots \times \R^{n_b} \quad\text{with}\quad n = n_1 + \cdots + n_b.
\end{equation}
For each $i \in \{1, \cdots, b\}$, we define a matrix $\Ubf_i \in \R^{n \times n_i}$ that injects a vector in $\R^{n_i}$ into $\R^n$ and its transpose $\Ubf_i^\Tsf$ that extracts the $i$th block from a vector in $\R^n$. For any $\xbm \in \R^n$, we have
\begin{equation}
\label{Eq:BlockExtraction}
\xbm = \sum_{i = 1}^b \Ubf_i\xbm_i \quad\text{with}\quad \xbm_i = \Ubf_i^\Tsf\xbm \in \R^{n_i}, \quad i = 1, \cdots, b \quad\Leftrightarrow\quad \sum_{i = 1}^b \Ubf_i\Ubf_i^\Tsf = \Ibf. 
\end{equation}
Note that~\eqref{Eq:BlockExtraction} directly implies the norm preservation $\|\xbm\|_2^2 = \|\xbm_1\|_2^2 + \cdots + \|\xbm_b\|_2^2$ for any $\xbm \in \R^n$. We are interested in a block-coordinate algorithm that uses only a subset of operator outputs corresponding to coordinates in some block $i \in \{1, \cdots, b\}$. Hence, for an operator $\Gsf: \R^n \rightarrow \R^n$, we define the block-coordinate operator $\Gsf_i: \R^n \rightarrow \R^{n_i}$ as
\begin{equation}
\Gsf_i(\xbm) \defn [\Gsf(\xbm)]_i = \Ubf_i^\Tsf \Gsf(\xbm) \in \R^{n_i}, \quad \xbm \in \R^n.
\end{equation}
We are in-particular interested in two operators: (a) the gradient $\nabla g(\xbm) = (\nabla_1 g(\xbm), \cdots, \nabla_b g(\xbm))$ of the data-fidelity term $g$ and (b) the denoiser $\Dsf_\sigmabm(\xbm) = (\Dsf_{\sigma_1}(\xbm_1), \cdots, \Dsf_{\sigma_b}(\xbm_b))$, where the vector $\sigmabm = (\sigma_1, \cdots, \sigma_b) \in \R_+^b$ consists of parameters for controling the strength of each block denoiser. Note how the denoiser acts in a separable fashion across different blocks.

\begin{algorithm}[t]
\caption{Block Coordinate Plug-and-Play Method (BC-PnP)}
\label{Alg:BCRED}
\begin{algorithmic}[1]
\STATE \textbf{input: } initial value $\xbm^0 \in \R^n$, parameters $\sigmabm \in \R_+^b$, and step-size $\gamma > 0$.
\FOR{$k = 1, 2, 3, \cdots$}
\STATE Choose an index $i_k \in \{1, \cdots, b\}$
\STATE $\xbm^k \leftarrow \xbm^{k-1} - \gamma \Ubf_{i_k} \Gsf_{i_k}(\xbm^{k-1})$\\ 
\quad where  $\Gsf_i(\xbm) \defn \Ubf_i^\Tsf\Gsf(\xbm)$ with $\Gsf(\xbm) \defn \frac{1}{\gamma}(\xbm - \Dsf_\sigmabm (\xbm - \gamma \nabla g(\xbm)))$.
\ENDFOR
\end{algorithmic}
\end{algorithm} 

\medskip\noindent
When $b = 1$, we have $\Ubf_1 = \Ubf_1^\Tsf = \Ibf$ and BC-PnP reduces to the conventional PnP-ISTA~\cite{Kamilov.etal2017, Xu.etal2020}. When $b > 1$, we have at least two blocks with BC-PnP updating only one block at a time
\begin{equation}
\label{Eq:BCPnPUpdate}
\xbm_j^k =
\begin{cases}
\xbm_j^{k-1} & \quad \text{when } j \neq i_k \\
\Dsf_{\sigma_j}(\xbm_j^{k-1}-\gamma \nabla_j g(\xbm^{k-1})) & \quad \text{when } j = i_k 
\end{cases},
\quad j \in \{1, \cdots, b\}.
\end{equation}

As with any coordinate descent method (see~\cite{Wright2015} for a review), BC-PnP can be implemented using different block selection strategies. One common strategy is to simply update blocks sequentially as $i_k = 1+\mod(k-1, b)$, where $\mod(\cdot)$ denotes the modulo operator. An alternative is to proceed in epochs of $b$ consecutive iterations, where at the start of each epoch the set $\{1, \cdots, b\}$ is reshuffled, and $i_k$ is then selected consecutively from this ordered set. Finally, one can adopt a fully randomized strategy where indices $i_k$ are selected as i.i.d.~random variables distributed uniformly over $\{1, \cdots, b\}$. 

Throughout this work, we will assume that each denoiser $\Dsf_{\sigma_i}$ is an \emph{approximate} MMSE estimator for the following AWGN denoising problem
\begin{equation}
\label{Eq:MMSEDenoisingProb}
\zbm_i = \xbm_i + \nbm_i \quad\text{with}\quad \xbm_i \sim p_{\xbm_i}, \quad \nbm_i \sim \mathcal{N}(0, \sigma_i^2 \Ibf),
\end{equation}
where $i \in \{1, \cdots, b\}$ and $\zbm_i \in \R^{n_i}$. We rely only on an \emph{approximation} of the MMSE estimator of $\xbm_i$ given $\zbm_i$, since the \emph{exact} MMSE denoiser corresponds to the generally intractable posterior mean
\begin{equation}
\label{Eq:MMSEDenoiser}
\Dsf_{\sigma_i}^\ast(\zbm_i) \defn \E[\xbm_i | \zbm_i] = \int_{\R^{n_i}} \xbm p_{\xbm_i | \zbm_i} (\xbm | \zbm_i) \d \xbm.
\end{equation}
Approximate MMSE denoisers are a useful model for denoisers due to the use of the MSE loss 
\begin{equation}
\label{Eq:MSELoss}
\Lcal(\Dsf_{\sigma_i}) = \E \left[ \|\xbm_i - \Dsf_{\sigma_i}(\zbm_i)\|_2^2 \right]
\end{equation}
for training deep denoisers, as well as the optimality of MMSE denoisers with respect to widely used image-quality metrics such as signal-to-noise ratio (SNR). 

As a simple illustration of the generality of BC-PnP, consider $b = 2$ with the least-squares objective
\begin{equation}
\label{Eq:BCDataFidelity}
g(\xbm) = \frac{1}{2}\|\ybm - \Abm(\thetabm)\vbm\|_2^2 \quad\text{with}\quad \xbm \defn (\vbm, \thetabm),
\end{equation}
where $\vbm \in \R^{n_1}$ denotes the unknown image and $\thetabm \in \R^{n_2}$ denotes the unknown parameters of the measurement operator. BC-PnP can then be implemented by first pre-training a dedicated AWGN denoiser $\Dsf_{\sigma_i}$ for each block $i$ and using it as a prior within Algorithm~\ref{Alg:BCRED}. It is also worth noting that the functions $g$ in~\eqref{Eq:BCDataFidelity} is nonconvex with respect to the variable $\xbm \in \R^n$. In the next section, we present the full convergence analysis of BC-PnP without any convexity assumptions on $g$ and nonexpansiveness assumptions on the denoiser $\Dsf_\sigmabm$.


\section{Convergence Analysis of BC-PnP}
\label{Sec:ConvergenceAnalysis}

In this section, we present two new theoretical convergence results for BC-PnP. We first discuss its convergence under the sequential updates and then under fully random updates. It is worth mentioning that BC-RED with fully random updates was theoretically analyzed in~\cite{Sun.etal2019b}. The novelty of our analysis here lies in that it allows for nonconvex functions $g$ and expansive denoisers $\Dsf_{\sigma_i}$. The nonconvexity of $g$ is essential since most data-fidelity terms used for blind inverse problems are nonconvex. On the other hand, by allowing expansive $\Dsf_{\sigma_i}$, our analysis avoids the need for the spectral normalization techniques that were previously suggested for PnP methods~\cite{Ryu.etal2019, Sun.etal2019b}. 

\medskip\noindent 
In the following, we will denote as $\Dsf_{\sigmabm}^\ast \defn (\Dsf^\ast_{\sigma_1}, \cdots, \Dsf^\ast_{\sigma_b})$ the exact MMSE denoiser in~\eqref{Eq:MMSEDenoiser}. Our analysis will require five assumptions that will serve as sufficient conditions for our theorems.
\begin{assumption}
\label{As:NonDegenerate}
The blocks $\xbm_i$ are independent with non-degenerate priors $p_{\xbm_i}$ over $\R^{n_i}$.
\end{assumption}
As a reminder, a probability distribution $p_{\xbm_i}$ is \emph{degenerate} over $\R^{n_i}$, if it is supported on a space of lower dimensions than $n_i$. Assumption~\ref{As:NonDegenerate} is required for establishing an explicit link between the MMSE denoiser~\eqref{Eq:MMSEDenoiser} and the following regularizer (see also~\cite{Gribonval2011, Xu.etal2020} for additional background)
\begin{equation}
\label{Eq:FullReg}
h(\xbm) = \sum_{i = 1}^b h_i(\xbm_i), \quad \xbm = (\xbm_1, \cdots, \xbm_n) \in \R^n,
\end{equation}
where each function $h_i$ is defined as (see the derivation in Section~\ref{Sup:Sec:MMSEBackground} of the supplement)
\begin{equation}
\label{Eq:ExpReg}
h_i (\xbm_i) \defn
\begin{cases}
-\frac{1}{2\gamma}\|\xbm_i - (\Dsf_{\sigma_i}^\ast)^{-1}(\xbm_i)\|_2^2 + \frac{\sigma_i^2}{\gamma} h_{\sigma_i}((\Dsf^\ast_{\sigma_i})^{-1}(\xbm_i)) & \text{for } \xbm_i \in \Im(\Dsf_{\sigma_i}^\ast) \\
+\infty & \text{for } \xbm_i \notin \Im(\Dsf_{\sigma_i}^\ast),
\end{cases}
\end{equation}
where $\gamma > 0$ is the step size, $(\Dsf_{\sigma_i}^\ast)^{-1}: \Im(\Dsf_{\sigma_i}^\ast) \rightarrow \R^{n_i}$ is the inverse mapping, which is well defined and smooth over $\Im(\Dsf_{\sigma_i}^\ast)$, and $h_{\sigma_i}(\cdot) \defn -\log(p_{\zbm_i}(\cdot))$, where $p_{\zbm_i}$ is the probability distribution over the AWGN corrupted observations~\eqref{Eq:MMSEDenoisingProb}. Note that the function $h_i$ is smooth for any $\xbm_i \in \Im(\Dsf_{\sigma_i}^\ast)$, which is the consequence of the smoothness of both $(\Dsf_{\sigma_i}^\ast)^{-1}$ and $h_{\sigma_i}$.

\begin{assumption}
\label{As:LipschitzDataFit}
The function $g$ is continuously differentiable and $\nabla g$ is Lipschitz continuous with constant $L > 0$. Additionally, each block gradient $\nabla_i g$ is block Lipschitz continuous with constant $L_i > 0$. We define the largest block Lipschitz constant as $\Lmax \defn \max\{L_1, \cdots, L_b\}$.
\end{assumption}
Lipschitz continuity of the gradient $\nabla g$ is a standard assumption in the context of imaging inverse problems. Note that we always have the relationship $(L/b) \leq \Lmax \leq L$ (see Section~3.2 in~\cite{Wright2015}).
\begin{assumption}
\label{As:BoundedFromBelow}
The explicit data-fidelity term and the implicit regularizer are bounded from below
\begin{equation}
\inf_{\xbm \in \R^n} g(\xbm) > -\infty, \quad \inf_{\xbm \in \R^n} h(\xbm) > -\infty.
\end{equation}
\end{assumption}
Assumption~\ref{As:BoundedFromBelow} implies that there exists $f^\ast > -\infty$ such that $f(\xbm) \geq f^\ast$ for all $\xbm \in \R^n$. Since Assumptions~\ref{As:NonDegenerate}-\ref{As:BoundedFromBelow} correspond to the standard assumptions used in the literature, they are broadly satisfied in the context of inverse problems.

\medskip\noindent
Our analysis assumes that at every iteration, BC-PnP uses inexact MMSE denoisers on each block. While there are several ways to specify the nature of ``inexactness,'' we consider the case where at every iteration $k$ of BC-PnP the distance of the output of $\Dsf_{\sigma_i}$ to $\Dsf^\ast_{\sigma_i}$ is bounded by a constant $\varepsilon_k$.
\begin{assumption}
\label{As:InexactDistance}
Each block denoiser $\Dsf_{\sigma_i}$ in $\Dsf_\sigmabm$ satisfies
\begin{equation*}
\|\Dsf_{\sigma_i}(\zbm_i^k)-\Dsf^\ast_{\sigma_i}(\zbm_i^k)\|_2 \leq \varepsilon_k, \quad i \in \{1, \cdots, b\}, \quad k = 1, 2, 3, \cdots,
\end{equation*}
where $\Dsf_{\sigma_i}^\ast$ is given in~\eqref{Eq:MMSEDenoiser} and $\zbm_i^k = \xbm_i^{k-1}-\gamma \nabla_i g(\xbm^{k-1})$.
\end{assumption}
For convenience, we will define quantities $\varepsilon^2 \defn \max\{\varepsilon^2_1, \varepsilon^2_2, \cdots\}$ and $\varepsilonbar_t^2 \defn (1/t) \left(\varepsilon_1^2 + \cdots + \varepsilon_t^2\right)$ that correspond to the largest and the mean squared-distances between the inexact and exact denoisers. Assumption~\ref{As:InexactDistance} states that the error of the approximate MMSE denoiser used for inference is bounded relative to the exact MMSE denoiser, which is reasonable when the approximate MMSE denoiser is a CNN trained to minimize the MSE.

\medskip\noindent
It has been shown in the prior work~\cite{Gribonval2011, Xu.etal2020} that the function $h$ in~\eqref{Eq:FullReg} is infinitely continuously differentiable over $\Im(\Dsf^\ast_\sigmabm)$. Our analysis requires the extension of the region where $h$ is smooth to include the range of the approximate MMSE denoiser, which is the goal of our next assumption.
\begin{assumption}
\label{As:LipschitzPrior}
Each regularizer $h_i$ in~\eqref{Eq:ExpReg} associated with the MMSE denoiser~\eqref{Eq:MMSEDenoiser} is continously differentiable and has a Lipschitz continuous gradient with constant $M_i > 0$ over the set
\begin{equation*}
\Im_\varepsilon(\Dsf^\ast_{\sigma_i}) \defn \{\xbm \in \R^{n_i} : \|\xbm-\Dsf_{\sigma_i}^\ast(\zbm)\|_2 \leq \varepsilon, \, \zbm \in \R^{n_i} \}, \quad i \in \{1, \cdots, b\}.
\end{equation*}
\end{assumption}
We will define $\Mmax \defn \max\{M_1, \cdots, M_b\}$ to be the largest Lipschitz constant and $\Im_\varepsilon(\Dsf^\ast_\sigmabm) \defn \{\xbm \in \R^n : \xbm_i \in \Im_\varepsilon(\Dsf^\ast_{\sigma_i}),\, i \in \{1, \dots, b\} \}$ to be the set over which $h$ is smooth. Assumption~\ref{As:LipschitzPrior} expands the region where the regularizer associated with the exact MMSE denoiser is smooth by including the range of the approximate MMSE denoiser. For example, this assumption is automatically true when the exact and approximate MMSE denoisers have the same range, which is reasonable when the approximate MMSE denoiser is trained to imitate the exact one.

\medskip\noindent
Our first theoretical result considers the sequential updates, where at each iteration, $i_k$ is selected as $i_k = 1+\mod(k-1, b)$ with $\mod(\cdot)$ being the modulo operator. We can then express any iterate $ib$ produced by BC-PnP for $i \geq 1$ as
\begin{equation*}
\xbm^{ib} = (\xbm_1^{ib}, \cdots, \xbm_b^{ib}) = (\xbm_1^{(i-1)b+1}, \cdots, \xbm_b^{ib}).
\end{equation*}
Note that $\xbm^{ib} \in \Im_\varepsilon(\Dsf^\ast_\sigmabm)$ since each block is an output of the denoiser. We prove the following result.
\begin{theorem}
\label{Thm:SeqConv}
Run BC-PnP under Assumptions~\ref{As:NonDegenerate}-\ref{As:LipschitzPrior} using the sequential block selection and the step $0 < \gamma < 1/\Lmax$. Then, we have
\begin{equation*}
\min_{1 \leq i \leq t} \|\nabla f(\xbm^{ib})\|_2^2 \leq  \frac{1}{t}\sum_{i = 1}^t \|\nabla f(\xbm^{ib}) \|_2^2 \leq \frac{C_1}{t}(f(\xbm^0)-f^\ast) + C_2 \varepsilonbar_{tb}^2,
\end{equation*}
where $C_1 > 0$ and $C_2 > 0$ are iteration independent constants. If additionally the sequence of error terms $\{\varepsilon_i\}_{i \geq 1}$ is square-summable, we have that $\nabla f(\xbm^{tb}) \rightarrow \zerobm$ as $t \rightarrow 0$.
\end{theorem}

Our second theoretical result considers fully random updates, where at each iteration, $i_k$ is selected as i.i.d.\ random variables distributed over $\{1, \cdots, b\}$. In this setting, we analyze the convergence of BC-PnP in terms of the sequence $\{\Gsf(\xbm^k)\}_{k \geq 0}$. Note that it is straightforward to verify that $\Zer(\Gsf) = \Zer(\nabla f)$, which makes this analysis meaningful. We prove the following result.
\begin{theorem}
\label{Thm:RandomConv}
Run BC-PnP under Assumptions~\ref{As:NonDegenerate}-\ref{As:LipschitzPrior} using the random i.i.d.\ block selection and the step $0 < \gamma < 1/\Lmax$. Then, we have
\begin{equation*}
\min_{1 \leq k \leq t} \E\left[\|\Gsf(\xbm^{k-1})\|_2^2\right] \leq \E \left[\frac{1}{t}\sum_{k = 1}^t \|\Gsf(\xbm^{k-1}) \|_2^2 \right] \leq \frac{D_1}{t}(f(\xbm^0)-f^\ast) + D_2 \varepsilonbar_t^2,
\end{equation*}
where $D_1 > 0$ and $D_2 > 0$ are iteration independent constants. If additionally the sequence of error terms $\{\varepsilon_i\}_{i \geq 1}$ is square-summable, we have that $\Gsf(\xbm^t) \asarrow \zerobm$ as $t \rightarrow \infty$.
\end{theorem}

The expressions for the constants in Theorems~\ref{Thm:SeqConv} and~\ref{Thm:RandomConv} are given in the proofs. The theorems show that if the sequence of approximation errors is square-summable, BC-PnP asymptotically achieves a stationary point of $f$. On the other hand, if the sequence of approximation errors is not square-summable, the convergence is only up to an error term that depends on the average of the squared approximation errors. Both theorems can thus be viewed as more flexible alternatives for the convergence analysis in~\cite{Sun.etal2019b}. 
It is also worth mentioning that the theorems are interesting even when the denoiser errors are not square-summable, since they provide explicit error bounds on convergence.
While the analysis in~\cite{Sun.etal2019b} assumes convex $g$ and nonexpansive $\Dsf_\sigmabm$, the analysis here does not require these two assumptions. It instead views the denoiser $\Dsf_\sigmabm$ as an \emph{approximation} to the MMSE estimator $\Dsf_\sigmabm^\ast$, where the approximation error is bounded by $\varepsilon_k$ at every iteration of BC-PnP. This view is compatible with denoisers trained to minimize the MSE loss~\eqref{Eq:MSELoss}. 

\begin{figure}[t]
    \centering
    \includegraphics[width=\textwidth]{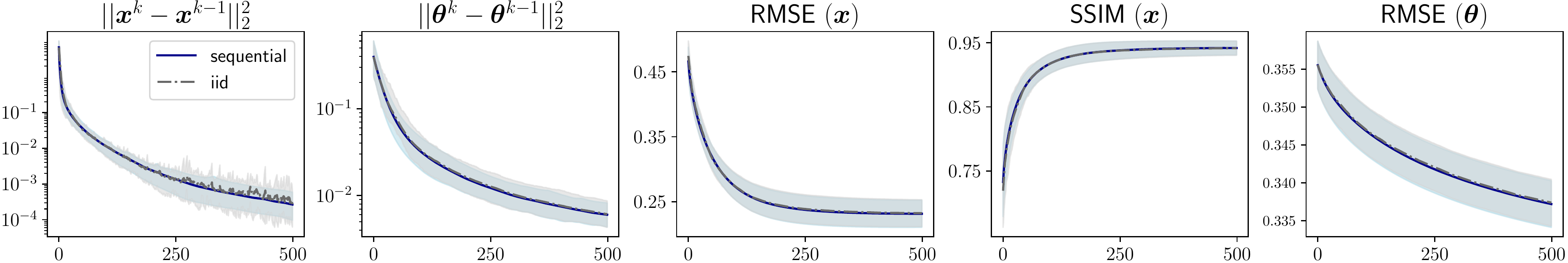}
    \vspace{-0.25cm}
    \caption{~\emph{Illustration of the BC-PnP convergence using the sequential and random i.i.d.\ block selection rules on CS-PMRI with the sampling factor $R=8$. Leftmost two plots: Evolution of the distance between two consecutive image and CSM iterates. Rightmost three plots: Evolution of the RMSE and SSIM metrics relative to the true solutions across BC-PnP iterations. Note how both block selection rules lead to a nearly identical convergence behaviour of BC-PnP in this experiment.}}
    \label{fig:exp_pmri_convergence}
\end{figure}

\section{Numerical Validation}

We numerically validate BC-PnP on two blind inverse problems: (a) \emph{compressed sensing parallel MRI (CS-PMRI)} with automatic coil sensitivity map (CSM) estimation and (b) \emph{blind image deblurring}. We adopt the traditional $\ell_2$-norm loss in~\eqref{Eq:BCDataFidelity} as the data-fidelity term for both problems. We will use $\xbm$ to denote the unknown image and $\thetabm$ to denote the unknown parameters of the measurement operator. We use the relative root mean squared error (RMSE) and structural similarity index (SSIM) as quantitative metrics to evaluate the performance.

We experimented with several ablated variants of BC-PnP, including PnP, PnP-GD$_\thetabm$, and PnP-oracle$_\thetabm$. 
PnP and PnP-oracle$_\thetabm$ denote basic variants of PnP that use pre-estimated and ground truth measurement operators, respectively. PnP-GD$_\thetabm$ is a variant of PnP based on~\cite{Xie.etal2021}, where $\thetabm$ is estimated without any DL prior. It is worth noting that PnP-oracle$_\thetabm$ is provided as an idealized reference in our experiment. As discussed in the following subsections, we also compare BC-PnP against several widely-used baseline methods specific to CS-PMRI and blind image deblurring.

\subsection{Compressed Sensing Parallel MRI}
The measurement operator of CS-PMRI consists of complex measurement operators $\Abm(\thetabm)\in\C^{m\times n}$ that depend on unknown CSMs $\{\thetabm_i\}$ in $\C^n$.
Each sub-measurement operator can be parameterized as $\Abm_i(\thetabm_i)=\Pbm\Fbm\diag(\thetabm_i)$, where $\Fbm$ is the Fourier transform, $\Pbm\in\R^{m\times n}$ is the sampling operator, and $\diag(\thetabm_i)$ forms a matrix by placing $\thetabm_i$ on its diagonal. We used T2-weighted MR brain acquisitions of 165 subjects obtained from the validation set of the fastMRI dataset~\cite{knoll2020fastmri} as the the fully sampled measurement for simulating measurements. We obtained reference $\thetabm_i$ from the fully sampled measurements using ESPIRiT~\cite{Uecker.etal2014}. These 165 subjects were split into 145, 10, and 10 for training, validation, and testing, respectively. BC-PnP and baseline methods were tested on 10 2D slices, randomly selected from the testing subjects. 
We followed ~\cite{knoll2020fastmri} to retrospectively undersample the fully sampled data using 1D Cartesian equispaced sampling masks with 10\% auto-calibration signal (ACS)~\cite{Uecker.etal2014} lines. We conducted our experiments for acceleration factors $R = 6$ and $8$. We adopted DRUNet~\cite{Zhang.etal2022} as the architectures of $\Dsf_\sigmabm$ for training both the image and CSM denoisers. BC-PnP and its ablated variants are initialized using CSMs $\thetabm_0$ pre-estimated using ESPIRiT\cite{Uecker.etal2014} and images $\xbm_0 \leftarrow \Abm(\thetabm_0)^\Hsf\ybm$, where $\Abm^\Hsf$ denotes the Hermitian transpose of $\Abm$.

We considered several baseline methods, including ENLIVE~\cite{Holme.etal2019}, ESPIRiT-TV~\cite{Uecker.etal2014}, Unet~\cite{Ronneberger.etal2015}, and ISTANet+~\cite{zhang2018ista}. ENLIVE is an iterative algorithm that jointly estimates images and coil sensitivity profiles. ESPIRiT-TV is an iterative algorithm that applies TV reconstruction method in~\eqref{equ:optimization}. Unet is trained to map raw measurements to desired ground truth without the knowledge of measurement operator. ISTANet+ denotes a widely-used DU architecture. We tested ESPIRiT-TV and ISTANet+ using CSMs pre-estimated using ESPIRiT. 

Figure~\ref{fig:exp_pmri_convergence} illustrates the convergence behaviour of BC-PnP on the test set for the acceleration factor $R=8$. Figure~\ref{fig:exp_pmri_image} illustrates reconstruction results for the acceleration factor $R=6$.  Table~\ref{tb:exp_pmri} summarizes the quantitative evaluation of BC-PnP relative to other PnP variants and the baseline methods. These results show that joint estimation can lead to significant improvements and that BC-PnP can perform as well as the idealized PnP-oracle$_\thetabm$ that knowns the true measurement operator.

\begin{figure}[t]
    \centering
    \includegraphics[width=.975\textwidth]{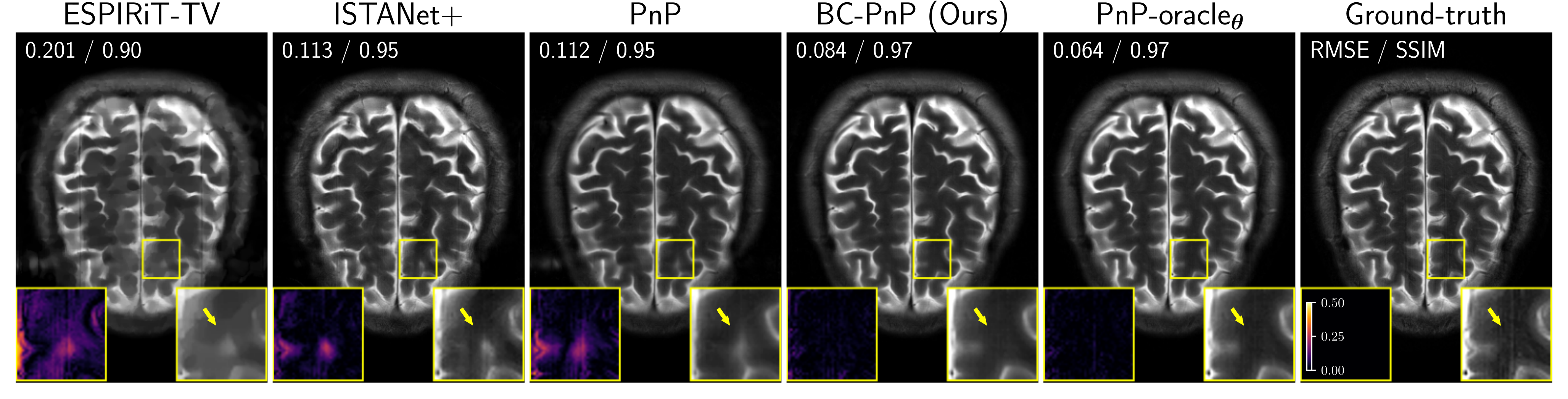}
         \caption{~\emph{Illustration of results from several well-known methods on CS-PMRI with the sampling factor $R = 6$. The quantities in the top-left corner of each image provide RMSE and SSIM values for each method. The squares at the bottom of each image visualize the error and the corresponding zoomed area in the image.  Note how BC-PnP using a deep denoiser on the unknown CSMs outperforms uncalibrated PnP and matches PnP-oracle$_\thetabm$ that knows the true CSMs.}}
    \label{fig:exp_pmri_image}
\end{figure}

\begin{table}[t]
\footnotesize
\centering
\renewcommand\arraystretch{0.75}
\setlength{\tabcolsep}{0.6pt}
\caption{RMSE and SSIM performance of several methods on CS-PMRI. The table highlights the {\color{lightred}\textbf{best}} and {\color{lightblue}\underline{second best}} results. The \textit{Calibration} column highlights methods specifically designed to solve the blind inverse problem. Note how the use of a DL prior over the measurement operator enables BC-PnP to outperform PnP and PnP-GD$_{\thetabm}$ and approach the performance of the oracle algorithm.}
\begin{threeparttable}
\begin{tabular}{lccccccccccc}
\toprule
\multirow{2}{*}{Method}               & & \multirow{2}{*}{\textit{Calibration (Y/N)}} &\ \ & \ & \multicolumn{3}{c}{$R=6$} & \ & \multicolumn{3}{c}{$R=8$} \\
\cmidrule{6-8}\cmidrule{10-12}
                     &   &   &   & \ & 
RMSE$_\xbm \downarrow$ & SSIM$_\xbm \uparrow$ & RMSE$_\thetabm \downarrow$ & \ \ & RMSE$_\xbm \downarrow$ & SSIM$_\xbm \uparrow$ & RMSE$_\thetabm \downarrow$
\\
\midrule
ENLIVE~\cite{Holme.etal2019}               &               &         \cmark      &                  &   & 0.371  & 0.763  & --- & \ & 0.419  & 0.730  & ---             \\
ESPIRiT-TV~\cite{Uecker.etal2014}           &               &     \cmark          &                  &   & 0.218  & 0.884  & 0.256 &  \ & 0.361  & 0.818  & 0.356              \\
Unet~\cite{Ronneberger.etal2015}                 &               &     \xmark          &                  &    & 0.218  & 0.904  & --- & \ & 0.195  & 0.907  & ---             \\
ISTANet+~\cite{zhang2018ista}            &               &       \xmark        &                  &   & {\color{lightblue}\underline{0.110}}  & 0.946  & --- & \ & {\color{lightblue}\underline{0.140}}  & {\color{lightblue}\underline{0.928}}  & ---              \\
\cmidrule{6-12}
PnP          &               &        \xmark       &                  &   & 0.111  & {\color{lightblue}\underline{0.950}}  & 0.256 & \ & 0.171  & 0.924  & 0.356              \\
PnP-GD$_\thetabm$~\cite{Xie.etal2021}         &               &       \cmark        &                  &   &  0.116 & {\color{lightblue}\underline{0.950}} & {\color{lightblue}\underline{0.254}}   & \ & 0.163 & 0.926 & {\color{lightblue}\underline{0.355}}             \\
\cmidrule{6-12}
BC-PnP (Ours)                 &               &        \cmark       &                  &    & {\color{lightred}\textbf{0.091}}  & {\color{lightred}\textbf{0.961}}  & {\color{lightred}\textbf{0.247}} & \ & {\color{lightred}\textbf{0.122}}  & {\color{lightred}\textbf{0.946}}  & {\color{lightred}\textbf{0.337}}             \\
\cmidrule{6-12}
PnP-oracle$_\thetabm$\tnote{$\star$}           &               &       \xmark        &                  &    & 0.069  & 0.969  & 0.000 & \ & 0.082  & 0.962  & 0.000             \\
\bottomrule
\end{tabular}
\begin{tablenotes}
\item[$\star$] \small{not available in practice for blind inverse problems.}
\end{tablenotes}
\end{threeparttable}
\label{tb:exp_pmri}
\end{table}

\begin{figure}[t]
    \centering
    \includegraphics[width=.975\textwidth]{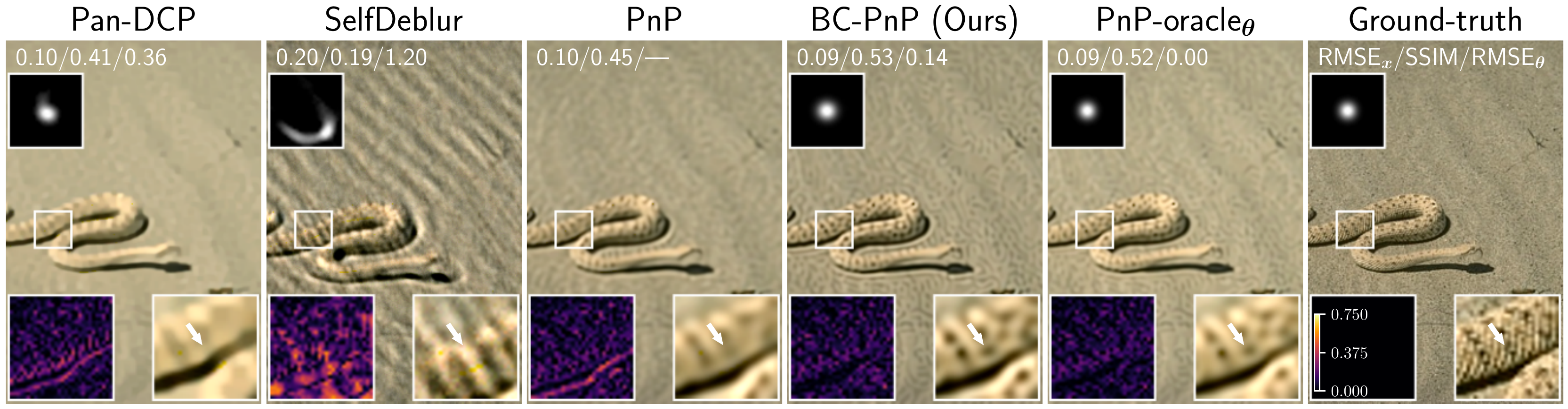}
    \caption{~\emph{Illustration of results from several well-known methods on blind image deblurring with the Gaussian kernel. The squares at the top of each image show the estimated kernels. The quantities in the top-left corner of each image provide RMSE and SSIM values for each method. The squares at the bottom of each image highlight the error and the corresponding zoomed image region. Note how the BC-PnP using a deep denoiser on the unknown kernel significantly outperforms the traditional PnP method and matches the performance of the oracle PnP method that knows the true blur kernel. Note also the effectiveness of BC-PnP for estimating the unknown blur kernel.}}
    \label{fig:exp_deconv_image}
\end{figure}

\begin{table}[t]
\footnotesize
\centering
\renewcommand\arraystretch{0.75}
\setlength{\tabcolsep}{0.6pt}
\caption{Quantitative evaluation of BC-PnP in blind image deblurring. We highlighted the {\color{lightred}\textbf{best}} and {\color{lightblue}\underline{second best}} results, respectively. The \textit{Calibration} column highlights methods specifically designed to solve the blind inverse problem. Note how the use of a prior over the measurement operator enables BC-PnP to nearly match the performance of the oracle algorithm.}

\begin{threeparttable}
\begin{tabular}{lccccccccccc}
\toprule
\multirow{2}{*}{Method}               & & \multirow{2}{*}{\textit{Calibration (Y/N)}} &\ \ & \ & \multicolumn{3}{c}{\includegraphics[width=0.035\textwidth]{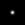}} & \ & \multicolumn{3}{c}{\includegraphics[width=0.035\textwidth]{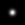}} \\
\cmidrule{6-8}\cmidrule{10-12}
                     &   &   &   & \ & 
RMSE$_\xbm \downarrow$ & SSIM$_\xbm \uparrow$ & RMSE$_\thetabm \downarrow$ & \ \ & RMSE$_\xbm \downarrow$ & SSIM$_\xbm \uparrow$ & RMSE$_\thetabm \downarrow$
\\
\midrule
Pan-DCP~\cite{pan2017deblurring}           &               &       \cmark        &                  &   &    0.087 & 0.835 & {\color{lightblue}\underline{0.283}}  &  \ & 0.114 & 0.733 & {\color{lightblue}\underline{0.246}}        \\
SelfDeblur~\cite{ren2020neural}               &               &    \cmark           &                  &   &   0.219 & 0.495 & 0.775 &  \ & 0.176 & 0.553 & 0.831       \\
DeblurGAN~\cite{kupyn2019deblurgan}                 &               &     \xmark          &                  &    &    0.090 & 0.823 & --- &  \ & 0.118 & 0.716 & ---   \\
USRNet~\cite{zhang.etal2020a}            &               &      \xmark         &                  &   &  0.106 & 0.855 & ---  &  \ & 0.114 & {\color{lightblue}\underline{0.769}} & ---         \\
\cmidrule{6-12}
PnP          &               &       \xmark        &                  &   &   {\color{lightblue}\underline{0.082}} & {\color{lightblue}\underline{0.857}} & {\color{lightblue}\underline{0.283}} &  \ & {\color{lightblue}\underline{0.106}} & 0.763 & {\color{lightblue}\underline{0.246}} \\
PnP-GD$_\thetabm$~\cite{Xie.etal2021}         &               &     \cmark          &                  &   &    {\color{lightblue}\underline{0.082}} & {\color{lightblue}\underline{0.857}} & {\color{lightblue}\underline{0.283}}  &  \ & 0.108 & 0.767 & {\color{lightblue}\underline{0.246}}    \\
\cmidrule{6-12}
BC-PnP (Ours)                 &               &       \cmark        &                  &    & {\color{lightred}\textbf{0.055}} & {\color{lightred}\textbf{0.921}} & {\color{lightred}\textbf{0.097}} &  \ & {\color{lightred}\textbf{0.098}} & {\color{lightred}\textbf{0.794}} & {\color{lightred}\textbf{0.107}} \\
\cmidrule{6-12}
PnP-oracle$_\thetabm$\tnote{$\star$}           &               &       \xmark        &                  &    &      0.051 & 0.929 & 0.000    &  \ & 0.088 & 0.817 & 0.000    \\
\bottomrule
\end{tabular}
\begin{tablenotes}
\item[$\star$] \small{not available in practice for blind inverse problems.}
\end{tablenotes}
\end{threeparttable}
\label{tb:exp_deconv}
\end{table}

\subsection{Blind Image Deblurring}

The measurement operator in blind image deblurring can be modeled as $\Abm(\thetabm)\xbm = \thetabm * \xbm$, where $\thetabm$ is the unknown blur kernel, $\xbm$ is the unknown image, and $*$ is the convolution. We randomly selected 10 testing ground truth image from CBSD68~\cite{Martin.etal2001} dataset. 
We generated $25\times 25$ Gaussian kernels with $\sigma=10$\footnote{We used \url{github.com/shangqigao/BayeSR} for generating the kernels.}. We tested the algorithms on 2 Gaussian kernels. We used a pre-trained image denoiser, as in the experimental setting of~\cite{Zhang.etal2022}. The kernel denoiser was trained on 10,000 generated kernels at several noise levels. 
We adopted DnCNN with 17 layers as the architectures of $\Dsf_\sigma$ for training kernel denoisers. BC-PnP and its ablated variants are initialized with the blur kernels $\thetabm_0$ pre-estimated using Pan-DCP~\cite{pan2017deblurring} and images $\xbm_0 \leftarrow \Abm(\thetabm_0)^\Tsf\ybm$.

We compared BC-PnP against several baseline methods, including Pan-DCP~\cite{pan2017deblurring}, SelfDeblur~\cite{ren2020neural}, DeblurGAN~\cite{kupyn2019deblurgan}, USRNet~\cite{zhang.etal2020a}. Pan-DCP is an optimization-based method that jointly estimates image and blur kernel. SelfDeblur trains two deep image priors (DIP)~\cite{Ulyanov.etal2018} to jointly estimate the blur kernel and the image. DeblurGAN is a supervised learning-based method that lacks the capability for kernel estimation, but can reconstruct images via direct inference. USRNet is a DU baseline that was tested using blur kernel estimated from~\cite{pan2017deblurring}. The results of DeblurGAN and USRNet are obtained by running the published code with the pre-trained weights.

Figure~\ref{fig:exp_deconv_image} illustrates the reconstruction results with a Gaussian kernel. Figure~\ref{fig:exp_deconv_image} demonstrates that BC-PnP can reconstruct the fine details of the snake skin, as highlighted by the white arrows, while both Pan-DCP and PnP produce smoother reconstructions.
Additionally, BC-PnP generates a more accurate blur kernel compared to the ground truth kernel, whereas Pan-DCP and SelfDeblur yield blur kernels with artifacts. Table~\ref{tb:exp_deconv} presents the quantitative evaluation of the reconstruction results using a Gaussian kernel, indicating that BC-PnP outperforms the baseline methods and nearly matches the SSIM and RMSE values of PnP-oracle$_\thetabm$ that is based on the ground truth blur kernel. 

\section{Conclusion}

The work presented in this paper proposes a new BC-PnP method for jointly estimating unknown images and unknown measurement operators in blind inverse problems, presents its theoretical analysis in terms of convergence and accuracy, and applies the method to two well-known blind inverse problems. The proposed method and its theoretical analysis extend the recent work on PnP by introducing a learned prior on the unknown measurement operator, dropping the convexity assumptions on the data-fidelity term, and nonexpansiveness assumptions on the denoiser. The numerical validation of BC-PnP shows the improvements due to the use of learned priors on the measurement operator and the ability of the method to match the performance of the oracle PnP method that knows the true measurement operator. One conclusion of this work is the potential effectiveness of PnP for solving inverse problems where the unknown quantities are not only images. 

\section*{Limitations}

The work presented in this paper comes with several limitations. The proposed BC-PnP method is based on PnP, which means that its performance is inherently limited by the use of AWGN denoisers as priors. While denoisers provide a convenient, principled, and flexible mechanism to specify priors, they are inherently self-supervised and their empirical performance can thus be suboptimal compared to priors trained in a supervised fashion for a specific inverse problem. PnP running over many iterations can also have higher computational complexity compared to some end-to-end alternatives, such as DU with a small number of steps. Our analysis is based on the assumption that the denoiser used for inference computes an approximation of the true MMSE denoiser. While this assumption is reasonable for deep denoisers trained using the MSE loss, it is not directly applicable to denoisers trained using other common loss functions, such as the $\ell_1$-norm or SSIM. Finally, as is common with most theoretical work, our analysis only holds when our assumptions are satisfied, which might limit its applicability in practice. Our future work will investigate ways to improve on the results presented here by exploring new PnP strategies for relaxing assumptions for convergence, considering end-to-end trained variants of BC-PnP based on DEQ~\cite{Gilton.etal2021, Liu.etal2022a}, and exploring BC-PnP using explicit regularizers~\cite{hurault2022proximal, Cohen.etal2021a, Hurault.etal2022}.

\section*{Broader Impact}

The expected impact of this work is in the area of imaging inverse problems with potential applications to computational imaging. There is a growing interest in computational imaging to leverage pre-trained deep models for estimating the unknown image as well as the unknown parameters of the imaging system. The ability to better address such problems can lead to new imaging tools for biomedical and scientific studies. While novel DL methods, such as the proposed BC-PnP approach, have the potential to enable new technological capabilities, they also come with a downside of being more complex and requiring higher-levels of technical sophistication. While our aim is to positively contribute to humanity, one can unfortunately envisage nonethical use of imaging technology.

\section*{Acknowledgments and Disclosure of Funding}
Research presented in this article was supported in part by the NSF CCF-2043134. This work is also supported by the NIH R01EB032713, RF1AG082030, RF1NS116565, and R21NS127425.

\newpage
\appendix

{\Large\textbf{Supplementary Material}}

\medskip

The mathematical analysis presented in this supplement builds on two distinct lines of work: (a) optimization-based characterization of the MMSE denoisers~\cite{Gribonval2011, Gribonval.Machart2013, Xu.etal2020}; (b) analysis of incremental optimization algorithms~\cite{Bertsekas2011, Bolte.etal2013, Wright2015, Mairal2015}. Our results are also related to two recent papers on PnP, namely the work on BC-RED in~\cite{Sun.etal2019b} and on PnP-ISTA in~\cite{Xu.etal2020}. Our results can be viewed as an extension of~\cite{Sun.etal2019b} to nonconvex data fidelity terms and expansive denoisers. They can also be viewed as an extension of~\cite{Xu.etal2020} to block-coordinate updates and possibly inexact MMSE denoisers.

\medskip\noindent
The structure of this supplementary document is as follows. In Section~\ref{Sup:Sec:SequentialUpdates}, we prove the convergence of BC-PnP under the deterministic sequential update rule. In Section~\ref{Sup:Sec:RandomUpdates}, we prove the convergence of BC-PnP under the random i.i.d.\ update rule. In Section~\ref{Sup:Sec:UsefulLemmas}, we provide technical lemmas useful for the proofs of the main theorems. In Section~\ref{Sup:Sec:BackgroundMaterial}, we provide background material useful for our theoretical analysis. In Section~\ref{Sup:Sec:AdditionalNumericalResults}, we provide additional simulations omitted from the main paper due to space.

\section{Proof of Theorem~\ref{Thm:SeqConv}}
\label{Sup:Sec:SequentialUpdates}

\begin{suptheorem}
Run BC-PnP under Assumptions~\ref{As:NonDegenerate}-\ref{As:LipschitzPrior} using the sequential block selection and the step $0 < \gamma < 1/\Lmax$. Then, we have
\begin{equation*}
\frac{1}{t}\sum_{i = 1}^t \|\nabla f(\xbm^{ib}) \|_2^2 \leq \frac{C_1}{t}(f(\xbm^0)-f^\ast) + C_2 \varepsilonbar_{tb}^2,
\end{equation*}
where $C_1 > 0$ and $C_2 > 0$ are iteration independent constants. If additionally the sequence of error terms $\{\varepsilon_i\}_{i \geq 1}$ is square-summable, we have that $\nabla f(\xbm^{tb}) \rightarrow \zerobm$ as $t \rightarrow 0$.
\end{suptheorem}

\begin{proof}
The block update $i \in \{1, \cdots, b\}$ of the sequential BC-PnP using the inexact and exact denoisers can be expressed
\begin{equation}
\xbm_i^{(k-1)b+i} =
\begin{cases}
\xbm_j^{(k-1)b+i-1} & \text{if } j \neq i \\
\Dsf_{\sigma_j}\left(\xbm_j^{(k-1)b+i-1}-\gamma \nabla_j g \left (\xbm^{(k-1)b+i-1}\right)\right) & \text{if } j = i
\end{cases}.
\end{equation}
and
\begin{equation}
\zbm_i^{(k-1)b+i} =
\begin{cases}
\xbm_j^{(k-1)b+i-1} & \text{if } j \neq i \\
\Dsf^\ast_{\sigma_j}\left(\xbm_j^{(k-1)b+i-1}-\gamma \nabla_j g \left (\xbm^{(k-1)b+i-1}\right)\right) & \text{if } j = i
\end{cases}.
\end{equation}
We introduce two variables
\begin{equation*}
\vbm^k \defn \xbm^{kb} = (\xbm_1^{(k-1)b+1}, \cdots, \xbm_b^{kb}) \quad\text{and}\quad \ubm^k \defn \zbm^{kb} = (\zbm_1^{(k-1)b+1}, \cdots, \zbm_b^{kb})
\end{equation*}
Since $h_i$ is smooth for any $\zbm_i \in \Im(\Dsf^\ast_{\sigma_i})$, the optimality conditions for each denoiser imply
\begin{equation}
\label{Eq:ProxOptimalityConditionInexact}
\nabla_i g\left(\xbm^{(k-1)b+i-1}\right) + \frac{1}{\gamma} \left(\zbm_i^{(k-1)b+i}-\xbm_i^{(k-1)b+i-1}\right) + \nabla h_i\left(\zbm_i^{(k-1)b+i}\right) = \zerobm,
\end{equation}
for each $i \in \{1, \cdots, b\}$ and any $k \geq 1$. Since we have
\begin{equation*}
\vbm^{k-1} = \xbm^{(k-1)b} = \left(\xbm_1^{(k-2)b+1}, \cdots, \xbm_b^{(k-1)b}\right) = \left(\xbm_1^{(k-1)b}, \cdots, \xbm_b^{kb-1}\right).
\end{equation*}
we can re-write~\eqref{Eq:ProxOptimalityConditionInexact} as
\begin{equation*}
\frac{1}{\gamma} \left(\vbm^{k-1}-\ubm^k \right) =
\begin{bmatrix}
\nabla_1 g\left(\xbm^{(k-1)b}\right) + \nabla h_1\left(\zbm_1^{(k-1)b+1}\right) \\[0.3em]
\vdots \\[0.3em]
\nabla_b g\left(\xbm^{kb-1}\right) + \nabla h_b\left(\zbm_b^{kb}\right)
\end{bmatrix}.
\end{equation*}
\medskip\noindent
From the smoothness of $h_i$ for $\zbm_i \in \Im(\Dsf^\ast_{\sigma_i})$, we thus have that
\begin{align*}
\nabla f(\ubm^k)
=
\begin{bmatrix}
\nabla_1 g(\ubm^k) + \nabla h_1\left(\zbm_1^{(k-1)b+1}\right) \\[0.3em]
\vdots \\[0.3em]
\nabla_b g(\ubm^k) + \nabla h_b\left(\zbm_b^{kb}\right)
\end{bmatrix} 
= \frac{1}{\gamma} (\vbm^{k-1}-\ubm^k) + 
\begin{bmatrix}
\nabla_1 g(\ubm^k) -\nabla_1 g(\xbm^{(k-1)b}) \\[0.3em]
\vdots \\[0.3em]
\nabla_b g(\ubm^k) - \nabla_b g(\xbm^{kb-1}).
\end{bmatrix}
\end{align*}
From the smoothness of $g$ and the sequential nature of updates, we can obtain the following bound
\begin{align*}
&\left\|
\begin{bmatrix}
\nabla_1 g(\ubm^k) -\nabla_1 g(\xbm^{(k-1)b}) \\[0.3em]
\vdots \\[0.3em]
\nabla_b g(\ubm^k) - \nabla_b g(\xbm^{kb-1}).
\end{bmatrix}
\right\|_2 
\leq 
\sum_{i = 1}^b \|\nabla_i g(\ubm^k)-\nabla_i g(\xbm^{(k-1)b-1+i})\|_2 \\
&\leq L \sum_{i = 1}^b \|\ubm^k - \xbm^{(k-1)b-1+i}\|_2 \leq b L \left( \|\ubm^k-\vbm^k\|_2+\|\vbm^k-\vbm^{k-1}\|_2 \right),
\end{align*}
where for the last inequality we used the triangualar inequality. By combining the last two equations and using the step-size $\gamma = 1/(\alpha \Lmax)$, we get
\begin{align*}
\|\nabla f(\ubm^k)\|_2 
&\leq \alpha \Lmax \|\ubm^k-\vbm^{k-1}\|_2 + bL \|\ubm^k-\vbm^k\|_2 + bL \|\vbm^k-\vbm^{k-1}\|_2 \\
&\leq (\alpha \Lmax+bL) \|\ubm^k-\vbm^k\|_2 + (\alpha\Lmax + bL) \|\vbm^k-\vbm^{k-1}\|_2 \\
&\leq (\alpha\Lmax + bL) \|\vbm^k-\vbm^{k-1}\|_2 + (\alpha\Lmax + bL) \sum_{i = 1}^b \varepsilon_{(k-1)b+i}.
\end{align*}
By using this bound, we can get the following bound for the iterate of BC-PnP
\begin{align*}
\|\nabla f(\vbm^k)\|_2 
&\leq \|\nabla f(\ubm^k)\|_2 + \|\nabla f(\vbm^k) - \nabla f(\ubm^k)\|_2 \\
&\leq \|\nabla f(\ubm^k)\|_2 + (L+\Mmax)\|\vbm^k-\ubm^k\|_2 \\
&\leq \|\nabla f(\ubm^k)\|_2 + (L+\Mmax) \sum_{i = 1}^b \varepsilon_{(k-1)b+i} \\
&\leq A_1 \|\vbm^k-\vbm^{k-1}\|_2 + A_2 \sum_{i = 1}^b \varepsilon_{(k-1)b+i}.
\end{align*}
with $A_1 \defn (\alpha\Lmax + bL)$ and $A_2 \defn (\alpha\Lmax + bL + L + \Mmax)$, where we first used the triangular inequality and then Lemma~\ref{Lem:LipRegInex}. By squaring both sides and using $(a+b)^2 \leq 2a^2 + 2b^2$
\begin{align*}
\|\nabla f(\vbm^k)\|_2^2 
&\leq 2A_1^2 \|\vbm^k-\vbm^{k-1}\|_2^2 + 2A_2^2 \left[\sum_{i = 1}^b \varepsilon_{(k-1)b+i}\right]^2 \\
&\leq 2A_1^2 \|\vbm^k-\vbm^{k-1}\|_2^2 + 2bA_2^2  \sum_{i = 1}^b \varepsilon_{(k-1)b+i}^2.
\end{align*}
By combining this inequality with Lemma~\ref{Lem:FullBlockUpdateIterInexact}, we get
\begin{equation*}
\|\nabla f(\vbm^k)\|_2^2 \leq B_1 (f(\vbm^{k-1})-f(\vbm^k)) + B_2 \sum_{i = 1}^b \varepsilon_{(k-1)b+i}^2
\end{equation*}
with $B_1 \defn 4 A_1^2/((\alpha-1)\Lmax)$ and $B_2 \defn 2bA_2^2 + \lambda A_1^2$, where $\lambda$ is given in Lemma~\ref{Lem:FullBlockUpdateIterInexact}. By averaging both sides of the bound over $t \geq 1$, we get the desired result
\begin{equation*}
\min_{1 \leq k \leq t} \|\nabla f(\vbm^k)\|_2^2 \leq  \frac{1}{t}\sum_{k = 1}^t \|\nabla f(\vbm^k)\|_2^2 \leq \frac{C_1}{t}(f(\xbm^0)-f^\ast) + C_2 \left[\frac{1}{tb}\sum_{k = 1}^{tb} \varepsilon_k^2\right]
\end{equation*}
where $C_1 \defn B_1$ and $C_2 \defn bB_2$.
\end{proof}

\section{Proof of Theorem~\ref{Thm:RandomConv}}
\label{Sup:Sec:RandomUpdates}

\begin{suptheorem}
Run BC-PnP under Assumptions~\ref{As:NonDegenerate}-\ref{As:LipschitzPrior} using the random i.i.d.\ block selection and the step $0 < \gamma < 1/\Lmax$. Then, we have
\begin{equation*}
\E \left[\frac{1}{t}\sum_{k = 1}^t \|\Gsf(\xbm^{k-1}) \|_2^2 \right] \leq \frac{D_1}{t}(f(\xbm^0)-f^\ast) + D_2 \varepsilonbar_t^2,
\end{equation*}
where $D_1 > 0$ and $D_2 > 0$ are iteration independent constants. If additionally the sequence of error terms $\{\varepsilon_i\}_{i \geq 1}$ is square-summable, we have that $\Gsf(\xbm^t) \asarrow \zerobm$ as $t \rightarrow \infty$.
\end{suptheorem}

\begin{proof}
To simplify our notations and analysis we will use $\gamma = 1/(\alpha \Lmax)$ with $\alpha > 1$. Note that Assumption~\ref{As:BoundedFromBelow} implies that there exists $f^\ast > -\infty$ such that we have almost surely $f^\ast \leq f(\xbm^k)$, $k \geq 1$.
Consider the iteration $k$ of BC-PnP in~\eqref{Eq:BCPnPUpdate}, where the random variables $i_k$ are selected uniformly at random from $\{1, \cdots, b\}$. 
This implies that
\begin{align}
\label{Eq:ExpBoundInex}
\E\left[\|\xbm^k-\xbm^{k-1}\|_2^2 \,|\, \xbm^{k-1}\right] 
&= \frac{1}{b}\sum_{j = 1}^b \|\xbm_j^{k-1}-\Dsf_{\sigma_j}(\xbm_j^{k-1}-\gamma \nabla_j g(\xbm^{k-1}))\|_2^2\\
\nonumber&= \frac{\gamma^2}{b}\sum_{j = 1}^b \|\Gsf_j(\xbm^{k-1})\|_2^2 = \frac{\gamma^2}{b}\|\Gsf(\xbm^{k-1})\|_2^2 = \frac{1}{b(\alpha\Lmax)^2} \|\Gsf(\xbm^{k-1})\|_2^2.
\end{align}
On the other hand, from Lemma~\ref{Lem:SingleUpdateInexact}, we have almost surely that
\begin{equation*}
f(\xbm^k) \leq f(\xbm^{k-1}) - (\alpha - 1)\frac{\Lmax}{2}\|\xbm_{i_k}^k-\xbm_{i_k}^{k-1}\|_2^2 + \frac{\lambda \varepsilon_k^2}{2}
\end{equation*}
By taking conditional expectation of this bound, subtracting $f^\ast$ from both sides, and using the equality~\eqref{Eq:ExpBoundInex}, we get
\begin{equation}
\label{Eq:PreSupermartInex}
\E\left[\left(f(\xbm^k)-f^\ast \right) \,|\, \xbm^{k-1} \right] \leq \left(f(\xbm^{k-1})-f^\ast\right) - \theta \|\Gsf(\xbm^{k-1})\|_2^2 + \frac{\lambda \varepsilon_k^2}{2},
\end{equation}
where $\theta \defn (\alpha-1)/(2\alpha^2 b \Lmax)$. Hence, by averaging over $t \geq 1$ iterations and taking the total expectation, we obtain
\begin{equation*}
\E\left[\frac{1}{t}\sum_{k = 1}^t \|\Gsf(\xbm^{k-1})\|_2^2\right] \leq \frac{D_1}{t} (f(\xbm^0)-f^\ast) + D_2 \left[\frac{1}{t}\sum_{k = 1}^t \varepsilon_k^2\right],
\end{equation*}
where $D_1 \defn 1/\theta$ and $D_2 \defn \lambda/(2\theta)$.
If $\{\varepsilon_k\}_{k \geq 1}$ in~\eqref{Eq:PreSupermartInex} is square summable, we can apply the supermartingale convergence theorem (see Section~\ref{Sup:Sec:BackgroundMaterial}), to get almost surely
\begin{equation*}
\sum_{k = 1}^\infty \|\Gsf(\xbm^{k-1})\|_2^2 < \infty,
\end{equation*}
which implies that $\|\Gsf(\xbm^k)\|_2 \asarrow 0$ as $k \rightarrow \infty$.

\end{proof}

\section{Useful technical lemmas}
\label{Sup:Sec:UsefulLemmas}

\begin{lemma}
\label{Lem:FullBlockUpdateIterInexact}
Consider the iteration $k \geq 1$ of BC-PnP under Assumptions~\ref{As:NonDegenerate}-\ref{As:LipschitzPrior} using the sequential block selection and the step-size $\gamma = 1/(\alpha\Lmax)$ with $\alpha > 1$. Then, we have that
\begin{equation*}
\|\vbm^k-\vbm^{k-1}\|_2^2 \leq \frac{2}{(\alpha-1)\Lmax}\left(f(\vbm^{k-1})-f(\vbm^k)\right) + \frac{\lambda}{2}\sum_{i = 1}^b \varepsilon_{(k-1)b+i}^2, \quad k \geq 1,
\end{equation*}
where $\vbm^k \defn (\xbm_1^{(k-1)b+1}, \cdots, \xbm_b^{kb})$, $f = g + h$ with $h$ defined in~\eqref{Eq:FullReg}, and $\lambda \defn \left(\alpha \Lmax + \Mmax\right)$.
\end{lemma}

\begin{proof}
First note that due to the sequential nature of block updates, we have that
\begin{align*}
&\vbm^k = \xbm^{kb} =  (\xbm_1^{(k-1)b+1}, \cdots, \xbm_b^{kb}),\\
&\vbm^{k-1} = \xbm^{(k-1)b} = \left(\xbm_1^{(k-2)b+1}, \cdots, \xbm_b^{(k-1)b}\right) = \left(\xbm_1^{(k-1)b}, \cdots, \xbm_b^{kb-1}\right).
\end{align*}
By combining the observation above with Lemma~\ref{Lem:SingleUpdateInexact}, we get the following bound
\begin{align*}
&(\alpha-1)\frac{\Lmax}{2}\sum_{i = 1}^b \|\xbm_i^{(k-1)b+i}-\xbm_i^{(k-1)b+i-1}\|_2^2 
= (\alpha-1)\frac{\Lmax}{2} \|\vbm^k-\vbm^{k-1}\|_2^2 \\
&\leq f(\xbm^{(k-1)b}) - f(\xbm^{kb})  + \frac{\lambda}{2}\sum_{i = 1}^b \varepsilon_{(k-1)b+i}^2= f(\vbm^{k-1}) - f(\vbm^k) + \frac{\lambda}{2}\sum_{i = 1}^b \varepsilon_{(k-1)b+i}^2,
\end{align*}
which directly leads to the desired result.
\end{proof}

\begin{lemma}
\label{Lem:SingleUpdateInexact}
Consider the iteration $k \geq 1$ of BC-PnP in~\eqref{Eq:BCPnPUpdate} with the step-size $\gamma = 1/(\alpha \Lmax)$ with $\alpha > 1$. If Assumptions~\ref{As:NonDegenerate}-\ref{As:LipschitzPrior} are true, we have that
\begin{equation}
f(\xbm^k) \leq f(\xbm^{k-1}) - (\alpha - 1)\frac{\Lmax}{2}\|\xbm_{i_k}^k-\xbm_{i_k}^{k-1}\|_2^2 + \frac{\lambda \varepsilon_k^2}{2},
\end{equation}
where $f = g + h$, with $h$ defined in~\eqref{Eq:FullReg}, and $\lambda \defn \left(\alpha \Lmax + \Mmax\right)$.
\end{lemma}

\begin{proof}
Consider the update that would be obtained by using the exact MMSE denoiser $\Dsf^\ast_{\sigma_i}$
\begin{equation*}
\xbm_j^\ast =
\begin{cases}
\xbm_j^{k-1} & \quad \text{when } j \neq i_k \\
\Dsf^\ast_{\sigma_j}(\xbm_j^{k-1}-\gamma \nabla_j g(\xbm^{k-1})) & \quad \text{when } j = i_k
\end{cases},
\quad j \in \{1, \cdots, b\},
\end{equation*}
From the fact that $\Dsf_{\sigma_i}^\ast$ is the MMSE denoiser, we know that $\xbm_{i_k}^\ast \in \Im_\varepsilon(\Dsf_{\sigma_{i_k}}^\ast)$ minimizes
\begin{equation}
\varphi_{i_k}(\ubm) \defn \frac{1}{2\gamma}\|\ubm-(\xbm_{i_k}^{k-1}-\gamma \nabla_{i_k} g(\xbm^{k-1})\|_2^2 + h_{i_k}(\ubm), \quad \ubm \in \R^{n_i}.
\end{equation}
From Assumption~\ref{As:LipschitzPrior}, we know that $\nabla h_{i_k}$ is $M_{i_k}$-Lipscthitz continuous over $\Im_\varepsilon(\Dsf^\ast_{\sigma_{i_k}})$, which implies
\begin{equation*}
\|\nabla \varphi_{i_k}(\ubm)-\nabla \varphi_{i_k}(\vbm)\|_2 \leq \lambda \|\ubm-\vbm\|_2 \quad \text{with} \quad \lambda \defn \left(\alpha \Lmax + \Mmax\right),
\end{equation*}
for all $\ubm, \vbm \in \Im_\varepsilon(\Dsf^\ast_{\sigma_{i_k}})$.
From the smoothness of $\varphi_{i_k}$ and since $\xbm_{i_k}^\ast$ minimizes it, we have that
\begin{align*}
\varphi_{i_k}(\xbm_{i_k}^k) 
\leq \varphi_{i_k}(\xbm_{i_k}^\ast) + \nabla \varphi_{i_k}(\xbm_{i_k}^\ast)^\Tsf(\xbm_{i_k}^k-\xbm_{i_k}^\ast) + \frac{\lambda}{2}\|\xbm_{i_k}^k-\xbm_{i_k}^\ast\|_2^2
\leq \varphi_{i_k}(\xbm_{i_k}^\ast) + \frac{\lambda \varepsilon_k^2}{2},
\end{align*}
where in the second inequality we used the bound on the denoisers in Assumption~\ref{As:InexactDistance}. We thus have
\begin{align*}
\varphi_{i_k}(\xbm_{i_k}^k) 
&= \frac{1}{2\gamma}\|\xbm_{i_k}^k-(\xbm_{i_k}^{k-1} - \gamma \nabla_{i_k} g(\xbm^{k-1}))\|_2^2 + h_{i_k}(\xbm_{i_k}^k)\\
&\leq \min_{\ubm \in \R^{n_i}} \left\{\frac{1}{2\gamma}\|\ubm-(\xbm_{i_k} - \gamma \nabla_{i_k} g(\xbm_{i_k}^{k-1}))\|_2^2 + h_{i_k}(\ubm)\right\} + \frac{\lambda \varepsilon_k^2}{2}\\
&\leq \frac{1}{2\gamma}\|\xbm_{i_k}^{k-1}-(\xbm_{i_k}^{k-1} - \gamma \nabla_{i_k} g(\xbm^{k-1}))\|_2^2 + h_{i_k}(\xbm_{i_k}^{k-1}) + \frac{\lambda\varepsilon_k^2}{2}.
\end{align*}
By expanding the first term on the left side of the inequality and simplifying, we obtain
\begin{equation}
\label{Eq:InexactRegBound}
h_{i_k}(\xbm_{i_k}^k) \leq h_i(\xbm_{i_k}^{k-1}) - \nabla_{i_k} g(\xbm^{k-1})^\Tsf(\xbm_{i_k}^k-\xbm_{i_k}^{k-1}) - \frac{1}{2\gamma}\|\xbm_{i_k}^k-\xbm_{i_k}^{k-1}\|_2^2 + \frac{\lambda \varepsilon_k^2}{2}.
\end{equation}
From the smoothness of $g$, we also have
\begin{equation}
\label{Eq:InexactDataFitBound}
g(\xbm^k) \leq g(\xbm^{k-1}) + \nabla_{i_k} g(\xbm^{k-1})^\Tsf(\xbm_{i_k}^k-\xbm_{i_k}^{k-1}) + \frac{\Lmax}{2}\|\xbm_{i_k}^k-\xbm_{i_k}^{k-1}\|_2^2.
\end{equation}
By combining~\eqref{Eq:InexactRegBound} and~\eqref{Eq:InexactDataFitBound}, and setting $\gamma = 1/(\alpha\Lmax)$, we get
\begin{align}
f(\xbm^k) 
&= g(\xbm^k) + h(\xbm^k) \\
&\leq g(\xbm^{k-1}) + \nabla_{i_k} g(\xbm^{k-1})^\Tsf(\xbm_{i_k}^k-\xbm_{i_k}^{k-1}) + \frac{\Lmax}{2}\|\xbm_{i_k}^k - \xbm_{i_k}^{k-1}\|_2^2 \\
&\quad + h(\xbm^{k-1}) - \nabla_{i_k} g(\xbm^{k-1})^\Tsf(\xbm_{i_k}^k-\xbm_{i_k}^{k-1}) - \frac{1}{2\gamma}\|\xbm_{i_k}^k - \xbm_{i_k}^{k-1}\|_2^2 + \frac{\lambda \varepsilon_k^2}{2}\\
&= f(\xbm^{k-1}) - (\alpha - 1)\frac{\Lmax}{2}\|\xbm_{i_k}^k-\xbm_{i_k}^{k-1}\|_2^2 + \frac{\lambda \varepsilon_k^2}{2},
\end{align}
where we used the fact that $\xbm_j^k = \xbm_j^{k-1}$ for all $j \neq i_k$.
\end{proof}

\begin{lemma}
\label{Lem:LipRegInex}
Suppose Assumptions~\ref{As:NonDegenerate},~\ref{As:InexactDistance}, and~\ref{As:LipschitzPrior} are true. We then have 
\begin{equation}
\|\nabla h(\xbm)-\nabla h(\zbm)\|_2 \leq \Mmax \|\xbm-\zbm\|_2, \quad \xbm, \zbm \in \Im_\varepsilon(\Dsf_\sigmabm^\ast).
\end{equation}
\end{lemma}
\begin{proof}
Consider $\xbm, \zbm \in \Im_\varepsilon(\Dsf_\sigmabm^\ast)$. From the $M_i$-Lipschitz continuity of each $\nabla h_i$, we get
\begin{align*}
\|\nabla h(\xbm)-\nabla h(\zbm)\|_2^2 
= \sum_{i = 1}^b \|\nabla h_i(\xbm_i) - \nabla h_i(\zbm_i)\|_2^2 \leq \sum_{i = 1}^b M_i^2 \|\xbm_i - \zbm_i\|_2^2 \leq \Mmax^2 \|\xbm-\zbm\|_2^2.
\end{align*}
\end{proof}

\section{Background material}
\label{Sup:Sec:BackgroundMaterial}

\subsection{Supermartingale convergence theorem}

Our analysis of the randomized BC-PnP algorithm relies on the classical result from the probability theory known as Supermargingale Convergence Theorem. The theorem is extensively used in the optimization literature (see Appendix A in~\cite{Ryu.Yin2023} and Proposition~2 in~\cite{Bertsekas2011}).

\begin{suptheorem}[\textbf{Supermartingale theorem}]
\label{Thm:SupermartingaleTheorem}
Let $F^k$, $G^k$, and $E^k$, be three sequences of random variables and let $\Fcal_k$ be sets of random variables such that $\Fcal_{k-1} \subseteq \Fcal_k$ for all $k \geq 1$. Assume that 
\begin{itemize}
\item $F^k$, $G^k$, and $E^k$ are functions of the random variables in $\Fcal_k$. Additionally, $F^k \geq 0$, $G^k \geq 0$, and $E^k \geq 0$ almost surely for $k \geq 1$.
\item For each $k \geq 1$, we have
\begin{equation*}
\E[F^k \,|\, \Fcal_{k-1}] \leq F^{k-1} - G^{k-1} + E^{k-1}.
\end{equation*}
\item We have almost surely
\begin{equation*}
\sum_{k = 0}^\infty E^k < \infty.
\end{equation*}
\end{itemize}
Then, we have almost surely
\begin{itemize}
\item $\sum_{k = 1}^\infty G^{k-1} < \infty$;
\item $F^k \rightarrow F^\infty$, where $F^\infty$ is a nonegative random variable.
\end{itemize}
\end{suptheorem}

\subsection{MMSE denoising as proximal operator}
\label{Sup:Sec:MMSEBackground}

The relationship between MMSE estimation and regularized inversion has been established by Gribonval in~\cite{Gribonval2011} and has been discussed in other contexts~\cite{Gribonval.Machart2013, Kazerouni.etal2013, Gribonval.Nikolova2021}. This relationship was formally connected to PnP methods in~\cite{Xu.etal2020}, leading to their new interpretation for MMSE denoisers. In this section, we review the key argument connecting MMSE denoising and proximal operators. 

It is well known that the estimator~\eqref{Eq:MMSEDenoiser} can be compactly expressed using the \emph{Tweedie's formula}
\begin{equation}
\label{Eq:Tweedie}
\Dsf^\ast_{\sigma_i}(\zbm_i) = \zbm_i - \sigma_i^2 \nabla h_{\sigma_i}(\zbm_i) \quad\text{with}\quad h_{\sigma_i}(\zbm_i) = -\log(p_{\zbm_i}(\zbm_i)),
\end{equation}
which can be obtained by differentiating~\eqref{Eq:MMSEDenoiser} using the expression for the probability distribution
\begin{equation}
\label{Eq:ConvRel}
p_{\zbm_i}(\zbm_i) = (p_{\xbm_i} \ast \phi_{\sigma_i})(\zbm_i) = \int_{\R^{n_i}} \phi_{\sigma_i}(\zbm_i-\xbm_i)p_{\xbm_i}(\xbm_i) \d \xbm_i,
\end{equation}
where
$$\phi_{\sigma_i}(\xbm_i) \defn \frac{1}{(2\pi\sigma_i^2)^{\frac{n_i}{2}}} \exp\left(-\frac{\|\xbm_i\|^2}{2\sigma_i^2}\right).$$
Since $\phi_{\sigma_i}$ is infinitely differentiable, so are $p_{\zbm_i}$ and $\Dsf^\ast_{\sigma_i}$. By differentiating $\Dsf^\ast_{\sigma_i}$, one can show that the Jacobian of $\Dsf_{\sigma_i}^\ast$ is positive definite (see Lemma 2 in~\cite{Gribonval2011})
\begin{equation}
\label{Eq:Jacobian}
\Jsf\Dsf_{\sigma_i}^\ast(\zbm_i) = \Ibf - {\sigma_i}^2 \Hsf h_{\sigma_i}(\zbm_i) \succ 0, \quad \zbm_i \in \R^{n_i},
\end{equation}
where $\Hsf h_{\sigma_i}$ denotes the Hessian matrix of the function $h_{\sigma_i}$. Finally, Assumption~\ref{As:NonDegenerate} also implies that $\Dsf_{\sigma_i}^\ast$ is a \emph{one-to-one} mapping from $\R^{n_i}$ to $\Im(\Dsf_{\sigma_i}^\ast)$, which means that ${(\Dsf_{\sigma_i}^\ast)^{-1}: \Im(\Dsf_{\sigma_i}^\ast) \rightarrow \R^{n_i}}$ is well defined and also infinitely differentiable over $\Im(\Dsf_{\sigma_i}^\ast)$ (see Lemma 1 in~\cite{Gribonval2011}). This directly implies that the regularizer $h_i$ in~\eqref{Eq:ExpReg} is also infinitely differentiable for any $\xbm_i \in \Im(\Dsf_{\sigma_i}^\ast)$.

We will now show that
\begin{align}
\label{Eq:DenoiserIsProx}
\Dsf_{\sigma_i}^\ast(\zbm_i) &= \prox_{\gamma h_i}(\zbm_i) = \argmin_{\xbm \in \R^{n_i}}\left\{\frac{1}{2}\|\xbm_i-\zbm_i\|^2 + \gamma h_i(\xbm_i)\right\}\nonumber
\end{align}
where $h_i$ is a (possibly nonconvex) function defined in~\eqref{Eq:ExpReg}. Our aim is to show that $\ubm^\ast = \zbm_i$ is the unique stationary point and global minimizer of
$$
\varphi(\ubm) \defn \frac{1}{2}\|\Dsf_{\sigma_i}^\ast(\ubm)-\zbm_i\|^2 + \gamma h_i(\Dsf_{\sigma_i}^\ast(\ubm)), \quad \ubm \in \R^{n_i}.
$$
By using the definition of $h_i$ in~\eqref{Eq:ExpReg} and the Tweedie's formula~\eqref{Eq:Tweedie}, we get
\begin{align*}
\varphi(\ubm) 
= \frac{1}{2}\|\Dsf_{\sigma_i}^\ast(\ubm)-\zbm_i\|^2 - \frac{\sigma_i^4}{2}\|\nabla h_{\sigma_i}(\ubm)\|^2 + \sigma_i^2h_{\sigma_i}(\ubm).
\end{align*}
The gradient of $\varphi$ is then given by
\begin{align*}
&\nabla \varphi(\ubm) 
= [\Jsf\Dsf^\ast_{\sigma_i}(\ubm)](\Dsf_{\sigma_i}^\ast(\ubm)-\zbm_i) + \sigma_i^2 [\Ibf - \sigma_i^2 \Hsf h_{\sigma_i}(\ubm)]\nabla h_{\sigma_i}(\ubm) = [\Jsf \Dsf_{\sigma_i}^\ast(\ubm)](\ubm-\zbm_i),
\end{align*}
where we used~\eqref{Eq:Jacobian} in the second line and~\eqref{Eq:Tweedie} in the third line. Now consider a scalar function ${q(\nu) = \varphi(\zbm_i+\nu\ubm)}$ and its derivative
$$q'(\nu) = \nabla \varphi(\zbm_i+\nu\ubm)^\Tsf\ubm = \nu \ubm^\Tsf [\Jsf\Dsf_{\sigma_i}^\ast(\zbm_i+\nu\ubm)]\ubm.$$
Positive definiteness of the Jacobian~\eqref{Eq:Jacobian} implies $q'(\nu) < 0$ and $q'(\nu) > 0$ for $\nu < 0$ and $\nu > 0$, respectively. Thus, $\nu = 0$ is the global minimizer of $q$. Since $\ubm \in \R^{n_i}$ is an arbitrary vector, we have that $\varphi$ has no stationary point beyond $\ubm^\ast = \zbm_i$ and that $\varphi(\zbm_i) < \varphi(\ubm)$ for any $\ubm \neq \zbm_i$.

\begin{figure}[!t]
    \centering
    \includegraphics[width=\textwidth]{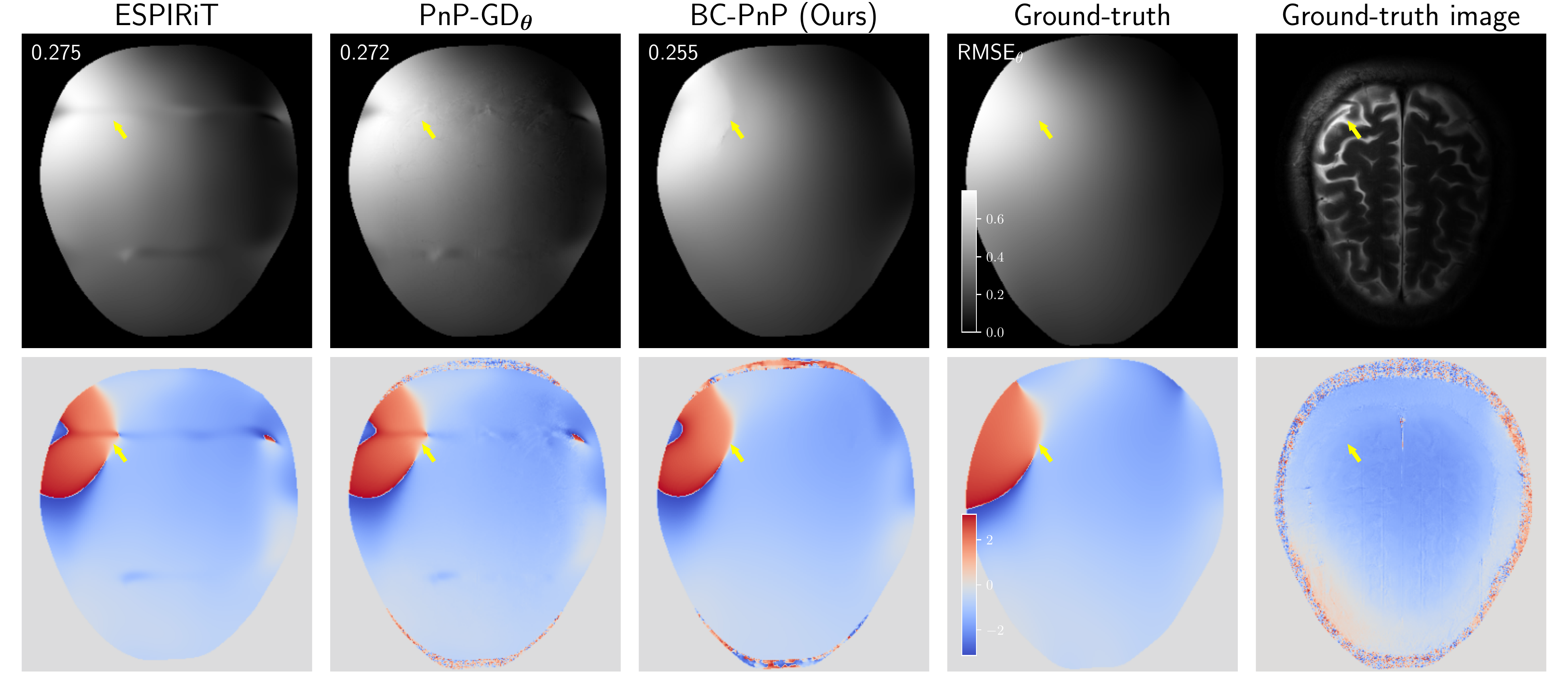}
    \caption{~\emph{Illustration of estimated CSM from several methods on CS-PMRI with the sampling factor $R=6$. The top and the bottom rows are the magnitude and the phase of the CSMs, respectively. The quantities in the top-left corner of each image in the top row provide RMSE values for each method. Ground-truth image was obtained using the fully sampled data corresponding to the ground truth CSMs. This figure highlights the effectiveness of BC-PnP for estimating the measurement operator.}}
    \label{fig:exp_pmri_theta}
\end{figure}

\section{Additional Technical Details}
\label{Sup:Sec:AdditionalNumericalResults}

We present some technical details and results that were omitted from the main paper. We used the following root mean squared error (RMSE) for quantitatively comparing different algorithms
\begin{equation}
    \mathsf{RMSE}(\zbmhat, \zbm) = \frac{\norm{\zbmhat -\zbm}_2}{\norm{\zbm}_2}
\end{equation}
where $\zbmhat$ and $\zbm$ represents the estimation and ground truth respectively. We ran BC-PnP and its ablated variants using a maximum number of $500$ iterations with the stopping criterion measuring the relative norm difference between iterations to be less than $10^{-5}$. 
We trained denoisers for images and measurement operators to optimize the MSE loss by using the Adam~\cite{Kingma.Ba2015} optimizer. We set the learning rate of Adam to $10^{-5}$.
We conducted all experiments on a machine equipped with an AMD Ryzen Threadripper 3960X 24-Core Processor and 4 NVIDIA GeForce RTX 3090 GPUs.

\subsection{Additional Details for CS-PMRI}

Figure~\ref{fig:exp_pmri_theta} illustrates the visual results of the estimated CSM for an acceleration factor of $R=6$. The widely used ESPIRiT algorithm estimates CSM directly from the ACS of the raw measurement, leading to imaging artifacts highlighted by yellow arrows under a high acceleration factor. Although PnP-GD$_{\thetabm}$ can reduce such imaging artifacts by jointly estimating images and CSMs, its performance is suboptimal compared to BC-PnP. Figure~\ref{fig:exp_pmri_theta} shows the effectiveness and superior performance of BC-PnP in estimating CSMs, which we attribute to its ability to use a DL denoiser as the CSM prior.

Figure~\ref{fig:supp_exp_pmri} visually illustrates results from several well-known methods, including those were omitted from the main paper, on CS-PMRI with acceleration factors $R=8$ and $R=6$. ESPIRiT-TV leads to the loss of details due to the well-known ``staircasing effect''. While Unet can outperform ESPIRiT-TV by learning a prior end-to-end from a training dataset, its performance is suboptimal compared with ISTA-Net+ that incorporates the pre-estimated measurement operator into the network architecture. PnP and PnP-GD$_{\thetabm}$ use pre-trained DL denoiser as image priors, leading competitive performance against ISTA-Net+. Figure~\ref{fig:supp_exp_pmri} demonstrates that BC-PnP can achieve quantitatively and qualitatively superior performance over several baselines by jointly estimating images and CSMs.

Figure~\ref{fig:supp_fastmri_groundtruth} shows ground truth MR images corresponding to the fully-sampled data that was used to generate measurement on the CS-PMRI experiments.
\begin{figure}[t]
    \centering
    \includegraphics[width=\textwidth]{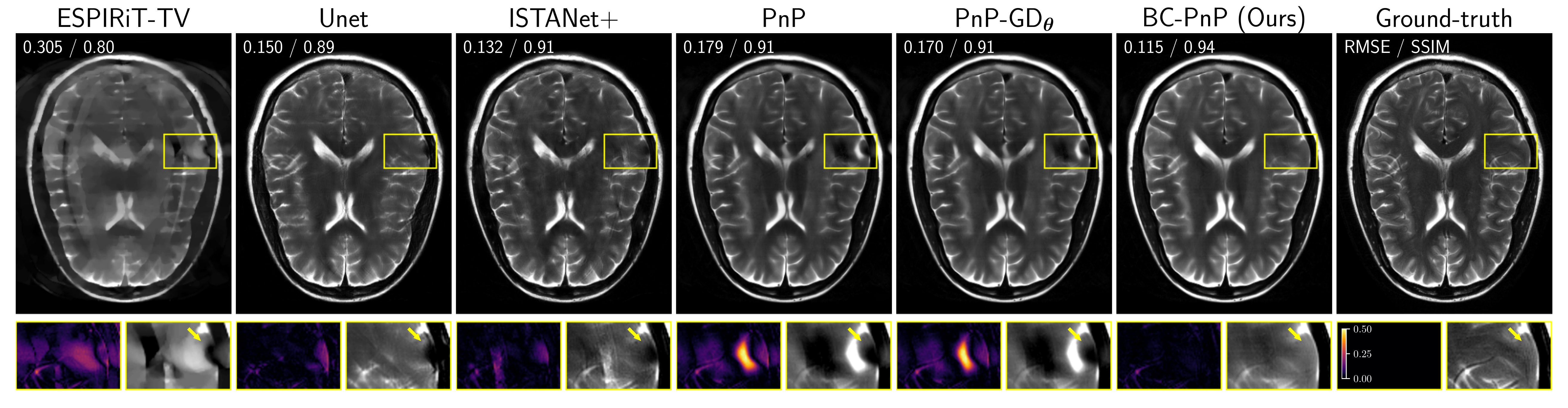}
    \includegraphics[width=\textwidth]{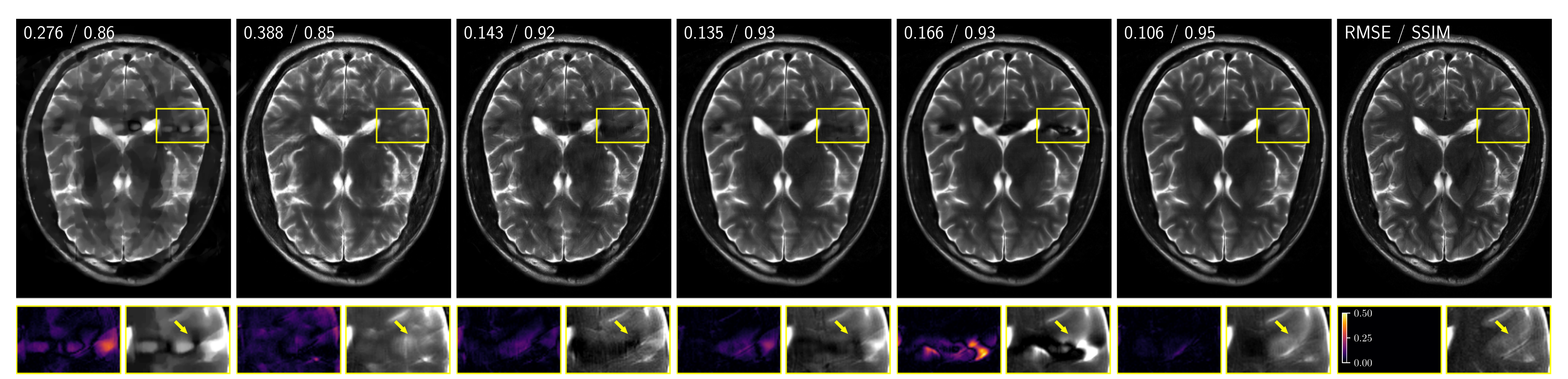}
    \caption{~\emph{Illustration of results from several well-known methods, including those were omitted from the main paper, on CS-PMRI with the sampling factor $R=8$ (\textbf{top row}) and $R=6$ (\textbf{bottom row}). The quantities in the top-left corner of each image provide the RMSE and SSIM values for each method. The squares at the bottom of each image shows the error and the corresponding zoomed area in the image. Note the excellent performance of BC-PnP that uses a learned deep denoiser on the CSMs.}} 
    \label{fig:supp_exp_pmri}
\end{figure}

\begin{figure}[t]
    \centering
    \includegraphics[width=\textwidth]{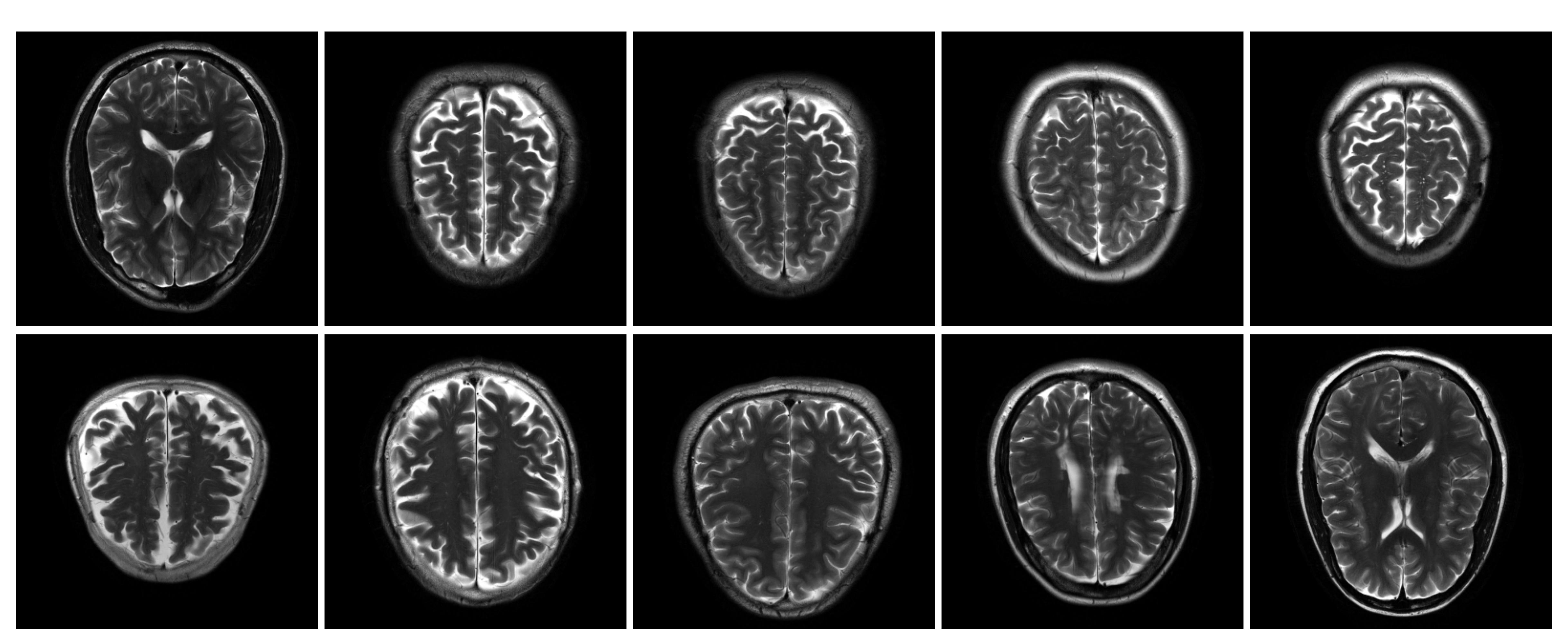}
    \caption{~\emph{Ground truth images that were used to generate the measurements in CS-PMRI.}}
    \label{fig:supp_fastmri_groundtruth}
\end{figure}

\subsection{Additional Details for Blind Image Deblurring}

Figure~\ref{fig:supp_exp_deconv_idx10} presents visual results from several well-known methods, including those were omitted from the main paper, on blind image deblurring with the Gaussian kernel. Pan-DCP estimates a deblur kernel from the blurry measurement and then reconstructs the image using a non-DL image prior. SelfDeblur jointly trains two deep image priors (DIPs) on the image and the blur kernel, respectively, but note how its reconstructions are translated compared to the ground truth. DeblurGANv2 enables debluring of an image without the knowledge of the blur kernel, but its performance is noticeably suboptimal. While USRNet reconstructs sharp images given a pre-estimated kernel, the details in the corresponding images are inconsistent relative to the ground truth (see the texture of the tiger skin highlighted by yellow arrows). Note how BC-PnP using a deep denoiser on the unknown kernel outperforms several baselines and matches the performance of PnP that knows the true kernel.

Figure~\ref{fig:supp_deconv_groundtruth} shows the images that were used to generate measurements for blind image deblurring.

\begin{figure}[t]
    \centering
    \includegraphics[width=\textwidth]{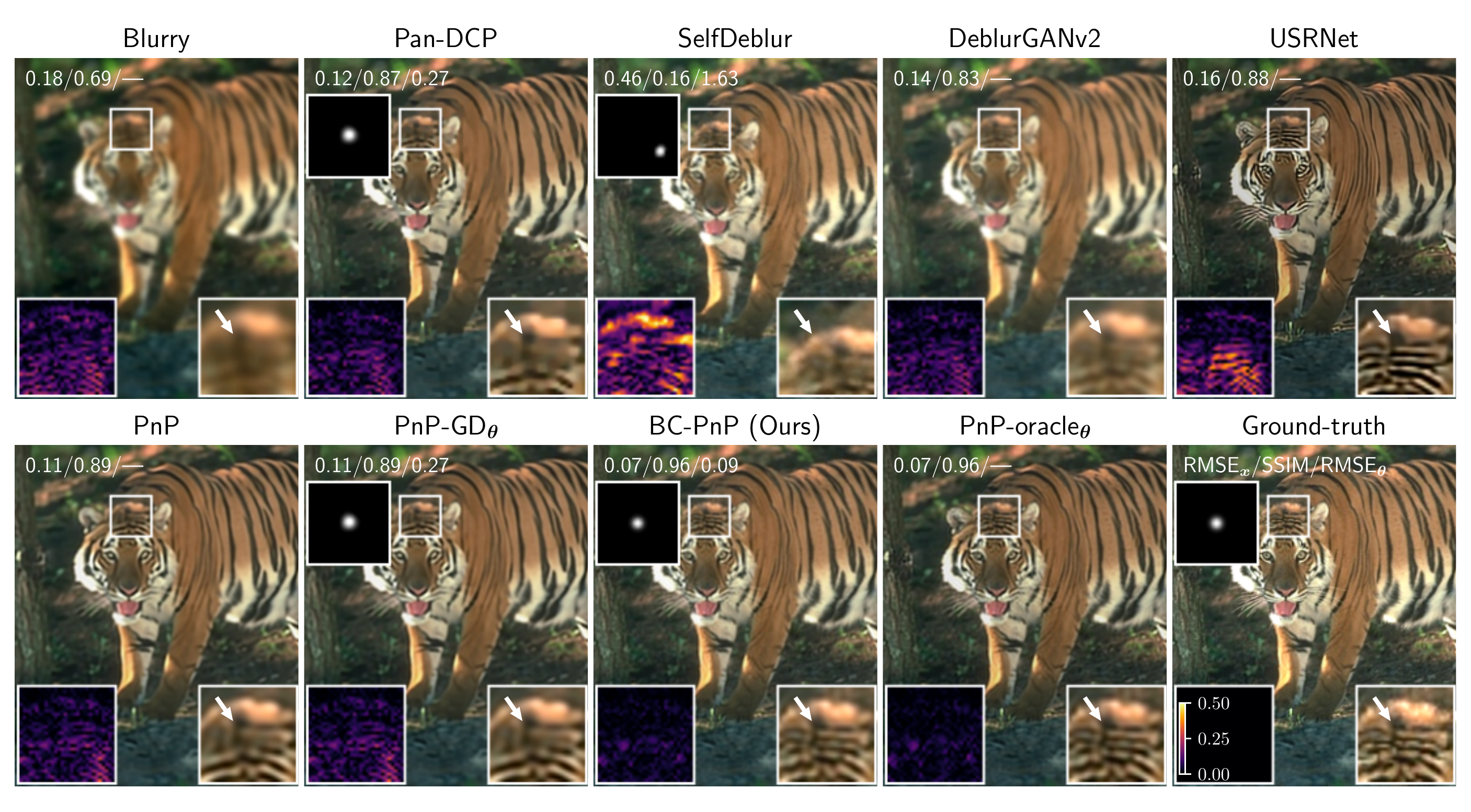}
    \caption{~\emph{Results from several well-known methods, including those were omitted from the main paper, on blind image deblurring with the Gaussian kernel. The squares at the top of each image show the estimated kernels. The quantities in the top-left corner of each image provide the RMSE and SSIM values for each method. The squares at the bottom of each image highlight the error and the corresponding zoomed image region. Note how BC-PnP using a deep denoiser on the unknown kernel performs as well as the oracle PnP that knows the ground truth kernel. Note also the effectiveness of BC-PnP for estimating the blur kernel.}}
    \label{fig:supp_exp_deconv_idx10}
\end{figure}

\begin{figure}[t]
    \centering
    \includegraphics[width=\textwidth]{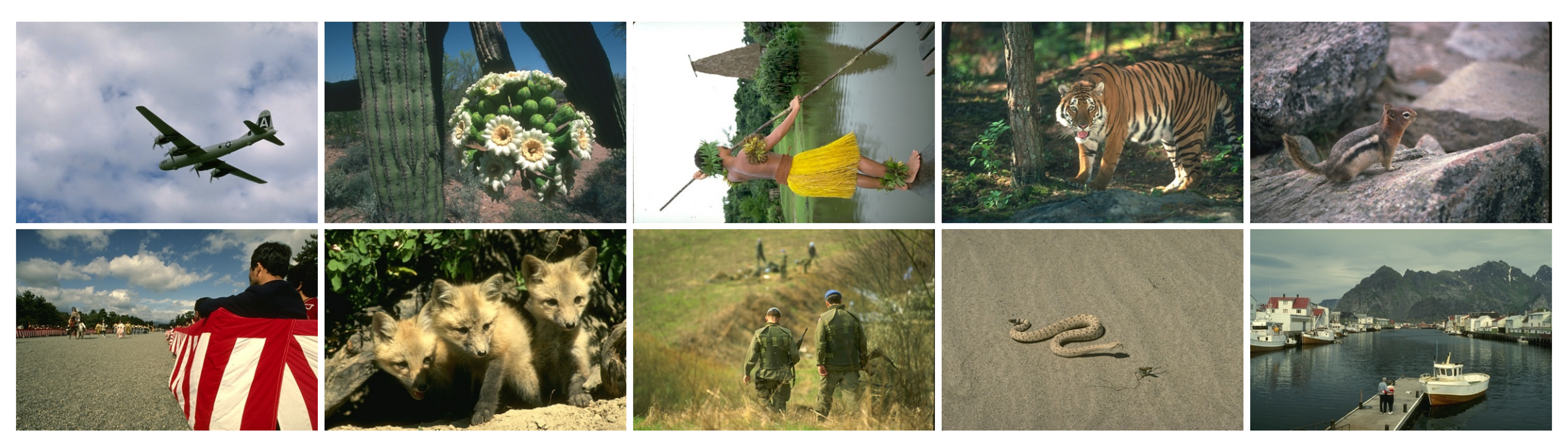}
    \caption{~\emph{Ground truth images used for generating measurements for blind image deblurring.}}
    \label{fig:supp_deconv_groundtruth}
\end{figure}


\begin{thebibliography}{10}

\bibitem{McCann.etal2017}
M.~T. McCann, K.~H. Jin, and M.~Unser,
\newblock ``Convolutional neural networks for inverse problems in imaging: A
  review,''
\newblock {\em IEEE Signal Process. Mag.}, vol. 34, no. 6, pp. 85--95, 2017.

\bibitem{Lucas.etal2018}
A.~Lucas, M.~Iliadis, R.~Molina, and A.~K. Katsaggelos,
\newblock ``Using deep neural networks for inverse problems in imaging:
  {B}eyond analytical methods,''
\newblock {\em IEEE Signal Process. Mag.}, vol. 35, no. 1, pp. 20--36, Jan.
  2018.

\bibitem{Ongie.etal2020}
G.~Ongie, A.~Jalal, C.~A. Metzler, R.~G. Baraniuk, A.~G. Dimakis, and
  R.~Willett,
\newblock ``Deep learning techniques for inverse problems in imaging,''
\newblock {\em IEEE J. Sel. Areas Inf. Theory}, vol. 1, no. 1, pp. 39--56, May
  2020.

\bibitem{Venkatakrishnan.etal2013}
S.~V. Venkatakrishnan, C.~A. Bouman, and B.~Wohlberg,
\newblock ``Plug-and-play priors for model based reconstruction,''
\newblock in {\em Proc. IEEE Global Conf. Signal Process. and Inf. Process.},
  Austin, TX, USA, Dec. 3-5, 2013, pp. 945--948.

\bibitem{Sreehari.etal2016}
S.~Sreehari, S.~V. Venkatakrishnan, B.~Wohlberg, G.~T. Buzzard, L.~F. Drummy,
  J.~P. Simmons, and C.~A. Bouman,
\newblock ``Plug-and-play priors for bright field electron tomography and
  sparse interpolation,''
\newblock {\em IEEE Trans. Comput. Imaging}, vol. 2, no. 4, pp. 408--423, Dec.
  2016.

\bibitem{Metzler.etal2018}
C.~Metzler, P.~Schniter, A.~Veeraraghavan, and R.~Baraniuk,
\newblock ``pr{D}eep: Robust phase retrieval with a flexible deep network,''
\newblock in {\em Proc. Int. Conf. Mach. Learn.}, Stockholmsm{\"a}ssan,
  Stockholm Sweden, Jul. 10--15 2018, pp. 3501--3510.

\bibitem{Zhang.etal2017a}
K.~Zhang, W.~Zuo, S.~Gu, and L.~Zhang,
\newblock ``Learning deep {CNN} denoiser prior for image restoration,''
\newblock in {\em Proc. IEEE Conf. Comput. Vis. Pattern Recognit.}, Honolulu,
  USA, July 21-26, 2017, pp. 3929--3938.

\bibitem{Meinhardt.etal2017}
T.~Meinhardt, M.~Moeller, C.~Hazirbas, and D.~Cremers,
\newblock ``Learning proximal operators: {U}sing denoising networks for
  regularizing inverse imaging problems,''
\newblock in {\em Proc. IEEE Int. Conf. Comp. Vis.}, Venice, Italy, Oct. 22-29,
  2017, pp. 1799--1808.

\bibitem{Dong.etal2019}
W.~{Dong}, P.~{Wang}, W.~{Yin}, G.~{Shi}, F.~{Wu}, and X.~{Lu},
\newblock ``Denoising prior driven deep neural network for image restoration,''
\newblock {\em IEEE Trans. Pattern Anal. Mach. Intell.}, vol. 41, no. 10, pp.
  2305--2318, Oct 2019.

\bibitem{Zhang.etal2019}
K.~Zhang, W.~Zuo, and L.~Zhang,
\newblock ``Deep plug-and-play super-resolution for arbitrary blur kernels,''
\newblock in {\em Proc. IEEE Conf. Comput. Vis. Pattern Recognit.}, Long Beach,
  CA, USA, June 16-20, 2019, pp. 1671--1681.

\bibitem{Wei.etal2020}
K.~Wei, A.~Aviles-Rivero, J.~Liang, Y.~Fu, C.-B. Sch\"onlieb, and H.~Huang,
\newblock ``Tuning-free plug-and-play proximal algorithm for inverse imaging
  problems,''
\newblock in {\em Proc. Int. Conf. Mach. Learn.}, 2020.

\bibitem{Zhang.etal2022}
K.~Zhang, Y.~Li, W.~Zuo, L.~Zhang, L.~{Van Gool}, and R.~Timofte,
\newblock ``Plug-and-play image restoration with deep denoiser prior,''
\newblock {\em IEEE Trans. Patt. Anal. and Machine Intell.}, 2022.

\bibitem{Liu.etal2022}
R.~Liu, Y.~Sun, J.~Zhu, L.~Tian, and U.~S. Kamilov,
\newblock ``Recovery of continuous {3D} refractive index maps from discrete
  intensity-only measurements using neural fields,''
\newblock {\em Nat. Mach. Intell.}, vol. 4, pp. 781--791, Sept. 2022.

\bibitem{Ahmad.etal2020}
R.~{Ahmad}, C.~A. {Bouman}, G.~T. {Buzzard}, S.~{Chan}, S.~{Liu}, E.~T.
  {Reehorst}, and P.~{Schniter},
\newblock ``Plug-and-play methods for magnetic resonance imaging: Using
  denoisers for image recovery,''
\newblock {\em IEEE Signal Process. Mag.}, vol. 37, no. 1, pp. 105--116, 2020.

\bibitem{Kamilov.etal2023}
U.~S. Kamilov, C.~A. Bouman, G.~T. Buzzard, and B.~Wohlberg,
\newblock ``Plug-and-play methods for integrating physical and learned models
  in computational imaging,''
\newblock {\em IEEE Signal Process. Mag.}, vol. 40, no. 1, pp. 85--97, Jan.
  2023.

\bibitem{Chan.etal2016}
S.~H. Chan, X.~Wang, and O.~A. Elgendy,
\newblock ``Plug-and-play {ADMM} for image restoration: Fixed-point convergence
  and applications,''
\newblock {\em IEEE Trans. Comp. Imag.}, vol. 3, no. 1, pp. 84--98, Mar. 2017.

\bibitem{Romano.etal2017}
Y.~Romano, M.~Elad, and P.~Milanfar,
\newblock ``The little engine that could: {R}egularization by denoising
  ({RED}),''
\newblock {\em SIAM J. Imaging Sci.}, vol. 10, no. 4, pp. 1804--1844, 2017.

\bibitem{Buzzard.etal2017}
G.~T. Buzzard, S.~H. Chan, S.~Sreehari, and C.~A. Bouman,
\newblock ``Plug-and-play unplugged: {O}ptimization free reconstruction using
  consensus equilibrium,''
\newblock {\em SIAM J. Imaging Sci.}, vol. 11, no. 3, pp. 2001--2020, Sep.
  2018.

\bibitem{Reehorst.Schniter2019}
E.~T. Reehorst and P.~Schniter,
\newblock ``Regularization by denoising: Clarifications and new
  interpretations,''
\newblock {\em IEEE Trans. Comput. Imag.}, vol. 5, no. 1, pp. 52--67, Mar.
  2019.

\bibitem{Sun.etal2018a}
Y.~Sun, B.~Wohlberg, and U.~S. Kamilov,
\newblock ``An online plug-and-play algorithm for regularized image
  reconstruction,''
\newblock {\em IEEE Trans. Comput. Imag.}, vol. 5, no. 3, pp. 395--408, Sept.
  2019.

\bibitem{Sun.etal2019b}
Y.~Sun, J.~Liu, and U.~S. Kamilov,
\newblock ``Block coordinate regularization by denoising,''
\newblock in {\em Proc. Adv. Neural Inf. Process. Syst.}, Vancouver, BC,
  Canada, Dec. 2019, pp. 382--392.

\bibitem{Ryu.etal2019}
Ernest~K. Ryu, J.~Liu, S.~Wang, X.~Chen, Z.~Wang, and W.~Yin,
\newblock ``Plug-and-play methods provably converge with properly trained
  denoisers,''
\newblock in {\em Proc. Int. Conf. Mach. Learn.}, Long Beach, CA, USA, Jun.
  09--15 2019, vol.~97, pp. 5546--5557.

\bibitem{Xu.etal2020}
X.~{Xu}, Y.~{Sun}, J.~{Liu}, B.~{Wohlberg}, and U.~S. {Kamilov},
\newblock ``Provable convergence of plug-and-play priors with mmse denoisers,''
\newblock {\em IEEE Signal Process. Lett.}, vol. 27, pp. 1280--1284, 2020.

\bibitem{Liu.etal2021b}
J.~Liu, S.~Asif, B.~Wohlberg, and U.~S. Kamilov,
\newblock ``Recovery analysis for plug-and-play priors using the restricted
  eigenvalue condition,''
\newblock in {\em Proc. Adv. Neural Inf. Process. Syst.}, December 6-14, 2021,
  pp. 5921--5933.

\bibitem{Kadkhodaie.Simoncelli2021}
Z.~Kadkhodaie and E.~P. Simoncelli,
\newblock ``Stochastic solutions for linear inverse problems using the prior
  implicit in a denoiser,''
\newblock in {\em Proc. Adv. Neural Inf. Process. Syst.}, December 6-14, 2021,
  pp. 13242--13254.

\bibitem{Cohen.etal2021a}
R.~Cohen, Y.~Blau, D.~Freedman, and E.~Rivlin,
\newblock ``It has potential: {G}radient-driven denoisers for convergent
  solutions to inverse problems,''
\newblock in {\em Proc. Adv. Neural Inf. Process. Syst.}, December 6-14, 2021,
  pp. 18152--18164.

\bibitem{Hurault.etal2022}
S.~Hurault, A.~Leclaire, and N.~Papadakis,
\newblock ``Gradient step denoiser for convergent plug-and-play,''
\newblock in {\em Proc. Int. Conf. Learn. Represent.}, 2022.

\bibitem{hurault2022proximal}
S.~Hurault, A.~Leclaire, and N.~Papadakis,
\newblock ``Proximal denoiser for convergent plug-and-play optimization with
  nonconvex regularization,''
\newblock in {\em Proc. Int. Conf. Mach. Learn.}, 2022, pp. 9483--9505.

\bibitem{Laumont.etal2022}
R.~Laumont, V.~{De Bortoli}, A.~Almansa, J.~Delon, A.~Durmus, and M.~Pereyra,
\newblock ``Bayesian imaging using plug \& play priors: {W}hen {L}angevin meets
  {T}weedie,''
\newblock {\em SIAM J. Imaging Sci.}, vol. 15, no. 2, pp. 701--737, 2022.

\bibitem{campisi2017blind}
P.~Campisi and K.~Egiazarian,
\newblock {\em Blind image deconvolution: theory and applications},
\newblock CRC press, 2017.

\bibitem{Fessler2020}
J.~A. Fessler,
\newblock ``Optimization methods for magnetic resonance image reconstruction,''
\newblock {\em IEEE Signal Process. Mag.}, vol. 1, no. 37, pp. 33--40, Jan.
  2020.

\bibitem{Ying.Sheng2007}
L.~Ying and J.~Sheng,
\newblock ``Joint image reconstruction and sensitivity estimation in {{SENSE}}
  ({{JSENSE}}),''
\newblock {\em Magn. Reson. Med.}, vol. 57, no. 6, pp. 1196--1202, June 2007.

\bibitem{Uecker.etal2008}
M.~Uecker, T.~Hohage, K.~T. Block, and J.~Frahm,
\newblock ``Image reconstruction by regularized nonlinear inversion-{{Joint}}
  estimation of coil sensitivities and image content,''
\newblock {\em Magn. Reson. Med.}, vol. 60, no. 3, pp. 674--682, Sept. 2008.

\bibitem{Jun.etal2021}
Y.~Jun, H.~Shin, T.~Eo, and D.~Hwang,
\newblock ``Joint deep model-based {{MR}} image and coil sensitivity
  reconstruction network (joint-icnet) for fast {{MRI}},''
\newblock in {\em Proc. {{IEEE Conf}}. {{Comput}}. {{Vis}}. {{Pattern
  Recognit}}.}, 2021, pp. 5270--5279.

\bibitem{Holme.etal2019}
H~C.~M Holme, S.~Rosenzweig, F.~Ong, R.~N Wilke, M.~Lustig, and M.~Uecker,
\newblock ``{{ENLIVE}}: An efficient nonlinear method for calibrationless and
  robust parallel imaging,''
\newblock {\em Scientific reports}, vol. 9, no. 1, pp. 1--13, 2019.

\bibitem{Arvinte.etal2021}
M.~Arvinte, S.~Vishwanath, A.~H. Tewfik, and J.~I. Tamir,
\newblock ``Deep {{J-Sense}}: {{Accelerated MRI}} reconstruction via unrolled
  alternating optimization,''
\newblock in {\em Proc. Med. Image. Comput. Comput. Assist. Intervent}, Apr.
  2021.

\bibitem{Sriram.etal2020}
A.~Sriram, J.~Zbontar, T.~Murrell, A.~Defazio, C~L. Zitnick, N.~Yakubova,
  F.~Knoll, and P.~Johnson,
\newblock ``End-to-end variational networks for accelerated {{MRI}}
  reconstruction,''
\newblock in {\em Proc. Med. Image. Comput. Comput. Assist. Intervent}, 2020,
  pp. 64--73.

\bibitem{kundur1996blind}
D.~Kundur and D.~Hatzinakos,
\newblock ``Blind image deconvolution,''
\newblock {\em IEEE Signal Process. Mag.}, vol. 13, no. 3, pp. 43--64, 1996.

\bibitem{pan2017deblurring}
J.~Pan, D.~Sun, H.~Pfister, and M.~Yang,
\newblock ``Deblurring images via dark channel prior,''
\newblock {\em IEEE Trans. Pattern Anal. Mach. Intell.}, vol. 40, no. 10, pp.
  2315--2328, 2017.

\bibitem{malhotra2016tomographic}
E.~Malhotra and A.~Rajwade,
\newblock ``Tomographic reconstruction from projections with unknown view
  angles exploiting moment-based relationships,''
\newblock in {\em Proc. IEEE Int. Conf. Image Proc.}, 2016, pp. 1759--1763.

\bibitem{lee2017phantomless}
D.~C Lee, P.~F Hoffmann, D.~L Kopperdahl, and T.~M Keaveny,
\newblock ``Phantomless calibration of ct scans for measurement of bmd and bone
  strength—inter-operator reanalysis precision,''
\newblock {\em Bone}, vol. 103, pp. 325--333, 2017.

\bibitem{cheng2018correction}
C.~Cheng, Y.~Ching, P.~Ko, and Y.~Hwu,
\newblock ``Correction of center of rotation and projection angle in
  synchrotron x-ray computed tomography,''
\newblock {\em Scientific reports}, vol. 8, no. 1, pp. 9884, 2018.

\bibitem{Xie.etal2021}
M.~Xie, J.~Liu, Y.~Sun, W.~Gan, B.~Wohlberg, and U.~S. Kamilov,
\newblock ``Joint reconstruction and calibration using regularization by
  denoising with application to computed tomography,''
\newblock in {\em Proc. {IEEE} Int. Conf. Comp. Vis. Workshops}, October 2021,
  pp. 4028--4037.

\bibitem{Wang2016.etal}
S.~{Wang}, Z.~{Su}, L.~{Ying}, X.~{Peng}, S.~{Zhu}, F.~{Liang}, D.~{Feng}, and
  D.~{Liang},
\newblock ``Accelerating magnetic resonance imaging via deep learning,''
\newblock in {\em Proc. Int. Symp. Biomedical Imaging}, April 2016, pp.
  514--517.

\bibitem{DJin.etal2017}
K.~H. {Jin}, M.~T. {McCann}, E.~{Froustey}, and M.~{Unser},
\newblock ``Deep convolutional neural network for inverse problems in
  imaging,''
\newblock {\em IEEE Trans. Image Process.}, vol. 26, no. 9, pp. 4509--4522,
  Sep. 2017.

\bibitem{Kang.etal2017}
E.~Kang, J.~Min, and J.~Ye,
\newblock ``A deep convolutional neural network using directional wavelets for
  low-dose x-ray {CT} reconstruction,''
\newblock {\em Medical Physics}, vol. 44, no. 10, pp. e360--e375, 2017.

\bibitem{Chen.etal2017}
H.~Chen, Y.~Zhang, M.~K. Kalra, F.~Lin, Y.~Chen, P.~Liao, J.~Zhou, and G.~Wang,
\newblock ``Low-dose {CT} with a residual encoder-decoder convolutional neural
  network,''
\newblock {\em IEEE Trans. Med. Imag.}, vol. 36, no. 12, pp. 2524--2535, Dec.
  2017.

\bibitem{Xu.etal2018a}
X.~Xu, J.~Pan, Y.~Zhang, and M.~Yang,
\newblock ``Motion {{Blur Kernel Estimation}} via {{Deep Learning}},''
\newblock {\em IEEE Trans. on Image Process.}, vol. 27, no. 1, pp. 194--205,
  Jan. 2018.

\bibitem{Monga.etal2021}
V.~Monga, Y.~Li, and Y.~C. Eldar,
\newblock ``Algorithm unrolling: {I}nterpretable, efficient deep learning for
  signal and image processing,''
\newblock {\em IEEE Signal Process. Mag.}, vol. 38, no. 2, pp. 18--44, Mar.
  2021.

\bibitem{Bora.etal2017}
A.~Bora, A.~Jalal, E.~Price, and A.~G. Dimakis,
\newblock ``Compressed sensing using generative priors,''
\newblock in {\em Proc. Int. Conf. Mach. Learn.}, Sydney, Australia, Aug. 2017,
  pp. 537--546.

\bibitem{zhang2018ista}
J.~{Zhang} and B.~{Ghanem},
\newblock ``{ISTA-Net}: {I}nterpretable optimization-inspired deep network for
  image compressive sensing,''
\newblock in {\em Proc. IEEE Conf. Comput. Vis. Pattern Recognit.}, 2018, pp.
  1828--1837.

\bibitem{Hauptmann.etal2018}
A.~{Hauptmann}, F.~{Lucka}, M.~{Betcke}, N.~{Huynh}, J.~{Adler}, B.~{Cox},
  P.~{Beard}, S.~{Ourselin}, and S.~{Arridge},
\newblock ``Model-based learning for accelerated, limited-view 3-d
  photoacoustic tomography,''
\newblock {\em IEEE Trans. Med. Imag.}, vol. 37, no. 6, pp. 1382--1393, 2018.

\bibitem{Gilton.etal2021}
D.~Gilton, G.~Ongie, and R.~Willett,
\newblock ``Deep equilibrium architectures for inverse problems in imaging,''
\newblock {\em IEEE Trans. Comput. Imag.}, vol. 7, pp. 1123--1133, 2021.

\bibitem{Liu.etal2022a}
J.~Liu, X.~Xu, W.~Gan, S.~Shoushtari, and U.~S. Kamilov,
\newblock ``Online deep equilibrium learning for regularization by denoising,''
\newblock in {\em Proc. Adv. Neural Inf. Process. Syst.}, New Orleans, LA,
  2022.

\bibitem{han.etal2018}
Y.~{Han} and J.~C. {Ye},
\newblock ``Framing {U-Net} via deep convolutional framelets: Application to
  sparse-view {CT},''
\newblock {\em IEEE Trans. Med. Imag.}, vol. 37, no. 6, pp. 1418--1429, 2018.

\bibitem{peng2022deepsense}
X.~Peng, B.~Sutton, F.~Lam, and Z.~Liang,
\newblock ``Deepsense: Learning coil sensitivity functions for sense
  reconstruction using deep learning,''
\newblock {\em Magn. Reson. Med.}, vol. 87, no. 4, pp. 1894--1902, 2022.

\bibitem{bostan2020deep}
E.~Bostan, R.~Heckel, M.~Chen, M.~Kellman, and L.~Waller,
\newblock ``Deep phase decoder: self-calibrating phase microscopy with an
  untrained deep neural network,''
\newblock {\em Optica}, vol. 7, no. 6, pp. 559--562, 2020.

\bibitem{ren2020neural}
D.~Ren, K.~Zhang, Q.~Wang, Q.~Hu, and W.~Zuo,
\newblock ``Neural blind deconvolution using deep priors,''
\newblock in {\em Proc. IEEE Conf. Comput. Vis. Pattern Recognit.}, 2020, pp.
  3341--3350.

\bibitem{asim2020blind}
M.~Asim, F.~Shamshad, and A.~Ahmed,
\newblock ``Blind image deconvolution using deep generative priors,''
\newblock {\em IEEE Trans. Comput. Imaging}, vol. 6, pp. 1493--1506, 2020.

\bibitem{Chung.etal2023a}
H.~Chung, J.~Kim, S.~Kim, and J.~Ye,
\newblock ``Parallel {{Diffusion Models}} of {{Operator}} and {{Image}} for
  {{Blind Inverse Problems}},''
\newblock in {\em Proc. {{IEEE Conf}}. {{Comput}}. {{Vis}}. {{Pattern
  Recognit}}.}, 2023.

\bibitem{hu2022spice}
Y.~Hu, W.~Gan, C.~Ying, T.~Wang, C.~Eldeniz, J.~Liu, Y.~Chen, H.~An, and U.~S
  Kamilov,
\newblock ``{SPICE}: Self-supervised learning for mri with automatic coil
  sensitivity estimation,''
\newblock {\em arXiv:2210.02584}, 2022.

\bibitem{Gossard.Weiss2022}
A.~Gossard and P.~Weiss,
\newblock ``Training {{Adaptive Reconstruction Networks}} for {{Blind Inverse
  Problems}},''
\newblock {\em arXiv:2202.11342}, 2022.

\bibitem{huang2022unrolled}
Y.~Huang, E.~Chouzenoux, and J.-C. Pesquet,
\newblock ``Unrolled variational bayesian algorithm for image blind
  deconvolution,''
\newblock {\em IEEE Trans. Image Process.}, vol. 32, pp. 430--445, 2022.

\bibitem{li2020efficient}
Y.~Li, M.~Tofighi, J.~Geng, V.~Monga, and Y.~C Eldar,
\newblock ``Efficient and interpretable deep blind image deblurring via
  algorithm unrolling,''
\newblock {\em IEEE Trans. Comput. Imaging}, vol. 6, pp. 666--681, 2020.

\bibitem{Kamilov.etal2017}
U.~S. Kamilov, H.~Mansour, and B.~Wohlberg,
\newblock ``A plug-and-play priors approach for solving nonlinear imaging
  inverse problems,''
\newblock {\em IEEE Signal. Proc. Let.}, vol. 24, no. 12, pp. 1872--1876, Dec.
  2017.

\bibitem{Gribonval2011}
R.~Gribonval,
\newblock ``Should penalized least squares regression be interpreted as maximum
  a posteriori estimation?,''
\newblock {\em IEEE Trans. Signal Process.}, vol. 59, no. 5, pp. 2405--2410,
  May 2011.

\bibitem{Tirer.Giryes2019}
T.~Tirer and R.~Giryes,
\newblock ``Image restoration by iterative denoising and backward
  projections,''
\newblock {\em IEEE Trans. Image Process.}, vol. 28, no. 3, pp. 1220--1234,
  Mar. 2019.

\bibitem{Teodoro.etal2019}
A.~M. Teodoro, J.~M. Bioucas-Dias, and {M. A. T.} Figueiredo,
\newblock ``A convergent image fusion algorithm using scene-adapted
  {G}aussian-mixture-based denoising,''
\newblock {\em IEEE Trans. Image Process.}, vol. 28, no. 1, pp. 451--463, Jan.
  2019.

\bibitem{Sun.etal2021}
Y.~Sun, Z.~Wu, B.~Wohlberg, and U.~S. Kamilov,
\newblock ``Scalable plug-and-play {ADMM} with convergence guarantees,''
\newblock {\em IEEE Trans. Comput. Imag.}, vol. 7, pp. 849--863, July 2021.

\bibitem{Cohen.etal2020}
R.~Cohen, M.~Elad, and P.~Milanfar,
\newblock ``Regularization by denoising via fixed-point projection (red-pro),''
\newblock {\em SIAM J. Imaging Sci.}, vol. 14, no. 3, pp. 1374--1406, 2021.

\bibitem{ljubenovic2017blind}
M.~Ljubenovi{\'c} and M.~A. Figueiredo,
\newblock ``Blind image deblurring using class-adapted image priors,''
\newblock in {\em Proc. IEEE Int. Conf. Image Proc.} IEEE, 2017, pp. 490--494.

\bibitem{ljubenovic2019plug}
M.~Ljubenovi{\'c} and M.~A. Figueiredo,
\newblock ``Plug-and-play approach to class-adapted blind image deblurring,''
\newblock {\em Int. J. Doc. Anal. Recognit.}, vol. 22, no. 2, pp. 79--97, 2019.

\bibitem{pellizzari2020coherent}
C.~J Pellizzari, M.~F Spencer, and C.~A Bouman,
\newblock ``Coherent plug-and-play: digital holographic imaging through
  atmospheric turbulence using model-based iterative reconstruction and
  convolutional neural networks,''
\newblock {\em IEEE Trans. Comput. Imaging}, vol. 6, pp. 1607--1621, 2020.

\bibitem{Wright2015}
S.~J. Wright,
\newblock ``Coordinate descent algorithms,''
\newblock {\em Math. Program.}, vol. 151, no. 1, pp. 3--34, June 2015.

\bibitem{knoll2020fastmri}
{F. {Knoll} \emph{et al.}},
\newblock ``{fastMRI}: A publicly available raw k-space and {DICOM} dataset of
  knee images for accelerated {MR} image reconstruction using machine
  learning,''
\newblock {\em Radiology: Artificial Intelligence}, vol. 2, no. 1, pp. e190007,
  2020.

\bibitem{Uecker.etal2014}
M.~Uecker, P.~Lai, M.~J. Murphy, P.~Virtue, M.~Elad, J.~M. Pauly, S.~S.
  Vasanawala, and M.~Lustig,
\newblock ``{ESPIRiT}---an eigenvalue approach to autocalibrating parallel mri:
  where sense meets grappa,''
\newblock {\em Magn. Reson. Med.}, vol. 71, no. 3, pp. 990--1001, 2014.

\bibitem{Ronneberger.etal2015}
O.~Ronneberger, P.~Fischer, and T.~Brox,
\newblock ``{U}-{N}et: {C}onvolutional networks for biomedical image
  segmentation,''
\newblock in {\em Proc. Med. Image. Comput. Comput. Assist. Intervent.}, 2015,
  pp. 234--241.

\bibitem{kupyn2019deblurgan}
O.~Kupyn, T.~Martyniuk, J.~Wu, and Z.~Wang,
\newblock ``Deblurgan-v2: Deblurring (orders-of-magnitude) faster and better,''
\newblock in {\em Proc. IEEE Int. Conf. Comput. Vis.}, 2019, pp. 8878--8887.

\bibitem{zhang.etal2020a}
K.~Zhang, L.~V. Gool, and R.~Timofte,
\newblock ``Deep unfolding network for image super-resolution,''
\newblock in {\em Proc. {IEEE} Conf. Comput. Vis. Pattern Recognit.}, Jun.
  2020, pp. 3217--3226.

\bibitem{Martin.etal2001}
D.~Martin, C.~Fowlkes, D.~Tal, and J.~Malik,
\newblock ``A database of human segmented natural images and its application to
  evaluating segmentation algorithms and measuring ecological statistics,''
\newblock in {\em Proc. {IEEE} Int. Conf. Comp. Vis.}, Vancouver, Canada, July
  7-14, 2001, pp. 416--423.

\bibitem{Ulyanov.etal2018}
D.~Ulyanov, A.~Vedaldi, and V.~Lempitsky,
\newblock ``Deep image prior,''
\newblock in {\em Proc. {IEEE} Conf. Computer Vision and Pattern Recognition},
  Salt Lake City, UT, USA, June 18-22, 2018, pp. 9446--9454.

\bibitem{Gribonval.Machart2013}
R.~Gribonval and P.~Machart,
\newblock ``Reconciling ``priors'' \& ``priors'' without prejudice?,''
\newblock in {\em Proc. Adv. Neural Inf. Process. Syst.}, Lake Tahoe, NV, USA,
  December 5-10, 2013, pp. 2193--2201.

\bibitem{Bertsekas2011}
D.~P. Bertsekas,
\newblock ``Incremental proximal methods for large scale convex optimization,''
\newblock {\em Math. Program. Ser. B}, vol. 129, pp. 163--195, 2011.

\bibitem{Bolte.etal2013}
J.~Bolte, S.~Sabach, and M.~Teboulle,
\newblock ``Proximal alternating linearized minimization for nonconvex and
  nonsmooth problems,''
\newblock {\em Math. Program.}, vol. 146, pp. 459--494, 2013.

\bibitem{Mairal2015}
J.~Mairal,
\newblock ``Incremental majorization-minimization optimization with application
  to large-scale machine learning,''
\newblock {\em SIAM J. Optim.}, vol. 25, no. 2, pp. 829--855, Jan. 2015.

\bibitem{Ryu.Yin2023}
E.~K. Ryu and W.~Yin,
\newblock {\em Large-Scale Convex Optimization: {A}lgorithms and Analyses via
  Monotone Operators},
\newblock Cambridge University Press, 2023.

\bibitem{Kazerouni.etal2013}
A.~Kazerouni, U.~S. Kamilov, E.~Bostan, and M.~Unser,
\newblock ``Bayesian denoising: From {MAP} to {MMSE} using consistent cycle
  spinning,''
\newblock {\em IEEE Signal Process. Lett.}, vol. 20, no. 3, pp. 249--252, March
  2013.

\bibitem{Gribonval.Nikolova2021}
R.~Gribonval and M.~Nikolova,
\newblock ``On {B}ayesian estimation and proximity operators,''
\newblock {\em Appl. Comput. Harmon. Anal.}, vol. 50, pp. 49--72, Jan. 2021.

\bibitem{Kingma.Ba2015}
D.~Kingma and J.~Ba,
\newblock ``Adam: {A} method for stochastic optimization,''
\newblock in {\em Proc. Int. Conf. on Learn. Represent.}, 2015.

\end{thebibliography}
\end{document}